\newif\ifdraft
\newif\ifanon
\newif\iftoc

\draftfalse
\anonfalse
\tocfalse

\ifdraft
 \documentclass[11pt,runningheads]{llncs}
  \usepackage[in]{fullpage}
\else
  \documentclass[11pt,runningheads]{llncs}
  \usepackage[in]{fullpage}
\fi

\iftoc
  \usepackage[nottoc]{tocbibind} %
\fi

\usepackage[T1]{fontenc}
\usepackage[utf8]{inputenc}
\usepackage{lmodern}
\usepackage{hyperref}
\hypersetup{
  pageanchor=false,
  pdfencoding=auto,         %
  bookmarksnumbered=true,   %
 colorlinks=true,          %
 linkcolor=blue,       %
 citecolor=blue,       %
 filecolor=blue,       %
 urlcolor=blue,        %
}
\ifdraft
  \usepackage{showlabels}
\fi
\usepackage[modulo]{lineno}
\ifdraft
  \linenumbers
\fi

\usepackage{amsthm}
\usepackage{amssymb,amsfonts}
\usepackage[autostyle=false, style=english]{csquotes}
\MakeOuterQuote{"}
\usepackage{amsmath}
\usepackage{mathtools}
\usepackage{algorithm,algorithmicx,algpseudocode}
\usepackage{paralist}
\usepackage[normalem]{ulem}
\usepackage{float}
\usepackage{xspace}
\usepackage{tcolorbox}
\usepackage{wrapfig}
\usepackage{enumitem}
\usepackage[
n,
advantage,
operators,
sets,
adversary,
landau,
probability,
notions,
logic,
ff,
mm,
primitives,
events,
complexity,
asymptotics,
keys,
]{cryptocode}
\usepackage[capitalise]{cleveref}
\usepackage{todonotes}

\makeatletter
\renewcommand\subsubsection{\@startsection{subsubsection}{3}{\z@}%
                       {-18\p@ \@plus -4\p@ \@minus -4\p@}%
                       {0.5em \@plus 0.22em \@minus 0.1em}%
                       {\normalfont\normalsize\bfseries\boldmath}}
\makeatother
\spnewtheorem{construction}{Construction}{\bfseries}{\itshape}
\Crefname{construction}{Construction}{Constructions}

\ifdraft
  \newcommand{\authnote}[3]{{\color{#3} {\bf #1:} #2}}
\else
  \newcommand{\authnote}[3]{}
\fi

\DeclareMathAlphabet{\mathpzc}{OT1}{pzc}{m}{it}

\newtheorem{fact}{Fact}

\makeatletter
\if@envcntsame
\let\c@remark\relax

\fi

\newcounter{game}%
\newcounter{algorithm saved}%

\makeatletter
\newenvironment{game}[1][htb]{%
    \renewcommand{\ALG@name}{Game}%
    \setcounter{algorithm saved}{\value{algorithm}} %
    \setcounter{algorithm}{\value{game}}%
    \begin{algorithm}[#1]%
    }{\end{algorithm}
    \setcounter{game}{\value{algorithm}}%
    \setcounter{algorithm}{\value{algorithm saved}}%
}
\crefname{algorithm}{Game}{Games}
\makeatother

\newcommand{\size}[1]{\left|#1\right|}

\newcommand{\dk}{\ensuremath{\mathsf{dk}}\xspace}

\newcommand{\gen}{\ensuremath{\mathsf{Gen}}\xspace}

\newcommand{\qpk}{\ensuremath{\mathpzc{qpk}}\xspace}
\newcommand{\qpkgen}{\ensuremath{\mathpzc{QPKGen}}\xspace}
\newcommand{\qenc}{\ensuremath{\mathpzc{Enc}}\xspace}

\newcommand{\qdec}{\ensuremath{\mathpzc{Dec}}\xspace}

\newcommand{\destruct}{\ensuremath{\mathpzc{Destruct}}\xspace}
\newcommand{\qc}{\ensuremath{\mathpzc{qc}} \xspace}

\newcommand{\PRFSPD}{\textsf{PRFSPD}}

\newcommand{\ver}{\ensuremath{\mathpzc{Ver}}\xspace}

\newcommand{\secproof}{\ensuremath{\mathsf{Unclonability\text{-}of\text{-}proofs}}\xspace}
\newcommand{\experiment}[2]{\ensuremath{\mathsf{Exp}^{#1,#2}_{\secpar}}\xspace}
\newcommand{\clon}[2]{\ensuremath{\mathsf{Cloning}\text{-}\experiment{#1}{#2}}\xspace}

\newcommand{\haar}{\ensuremath{\mathpzc{Haar}}\xspace}
\newcommand{\uniform}[1]{\ensuremath{\xleftarrow{u}\{0,1\}^{#1}}\xspace}

\newcommand{\bra}[1]{\langle #1 \vert}
\newcommand{\ket}[1]{\vert #1 \rangle}
\newcommand{\ketbra}[1]{\vert #1 \rangle \langle #1 \vert}
\newcommand{\braket}[2]{\langle #1 \vert #2 \rangle}

\newcommand{\tensor}{\otimes}

\newcommand{\keygen}{\ensuremath{\mathsf{key\text{-}gen}}}

\newcommand{\qpt}{\textsf{QPT}}

\newcommand{\ct}{\textsf{ct}}
\newcommand{\pt}{\textsf{pt}}
\newcommand{\Hilb}{\mathcal{H}}
\renewcommand{\CC}{\mathbb{C}}
\newcommand{\Gen}{\ensuremath{\mathpzc{Gen}}}

\newcommand{\PRF}{\textsf{PRF}}
\DeclareMathOperator{\trace}{Tr}

\newcommand{\ot}{\otimes}
\newcommand{\ident}{\mathcal{I}}

\renewcommand{\secpar}{\lambda}
\renewcommand{\secparam}{1^\secpar}
\newcommand{\ooracle}{\mathcal{O}}
\newcommand{\SE}{\mathsf{SE}}
\newcommand{\oenc}{\mathsf{OEnc}}
\newcommand{\odec}{\mathsf{ODec}} %
\ifdraft
\newcounter{mycomment}

\fi

\title{Public-Key Encryption with Quantum Keys}
\ifanon
  \author{}
  \institute{}
  \date{}
\else
    \author{Khashayar Barooti\inst{1}
      \and Alex B. Grilo\inst{2}
      \and Loïs Huguenin-Dumittan\inst{1}
      \and Giulio~Malavolta\inst{3}
      \and Or~Sattath\inst{4}
      \and Quoc-Huy~Vu\inst{2}
      \and Michael Walter\inst{5}}
    \institute{EPFL, Lausanne, Switzerland \and
      Sorbonne Universit\'e, CNRS, LIP6, France \and
      Max-Planck Institute in Security and Privacy, Bochum, Germany \and
      Computer Science Department, Ben-Gurion University of the Negev, Israel \and
      Faculty of Computer Science, Ruhr University Bochum, Germany}
    \authorrunning{K. Barooti et al.}
\fi

\begin{document}
\maketitle
\begin{abstract}
  In the framework of Impagliazzo's five worlds, a distinction is often made
  between two worlds, one where public-key encryption exists (Cryptomania), 
  and one in which only one-way functions exist (MiniCrypt).
  However, the boundaries between these worlds can change when quantum
  information is taken into account.
  Recent work has shown that quantum variants of oblivious transfer and
  multi-party computation, both primitives that are classically in Cryptomania,
  can be constructed from one-way functions, placing them in the realm of
  quantum MiniCrypt (the so-called MiniQCrypt).
  This naturally raises the following question: \emph{Is it possible to
    construct a quantum variant of public-key encryption, which is at the heart
    of Cryptomania, from one-way functions or potentially weaker assumptions?}

  In this work, we initiate the formal study of the notion of quantum public-key
  encryption (qPKE), i.e., public-key encryption where keys are allowed to be quantum states.
  We propose new definitions of security and several constructions of qPKE based
  on the existence of one-way functions (OWF), or even weaker assumptions, such as
 pseudorandom function-like states
  (PRFS) and pseudorandom function-like states
  with proof of destruction (PRFSPD).
  Finally, to give a tight characterization of this primitive, we show that
  computational assumptions are necessary to build quantum public-key
  encryption. That is, we give a self-contained proof that no quantum public-key encryption scheme can provide
  information-theoretic security.
\end{abstract}

\iftoc
    \tableofcontents 
\fi

\section{Introduction} %
\label{sec:introduction}

The use of quantum resources to enable cryptographic tasks under weaker
assumptions than classically needed (or even {\em unconditionally}) were
actually the first concrete proposals of quantum computing, with the seminal
quantum money protocol of Wiesner~\cite{Wie83} and the key-exchange protocol of
Bennett and Brassard~\cite{C:BenBra84}.  Ever since, the field of quantum cryptography has seen a surge of primitives that leverage quantum 
information to perform tasks that classically require stronger assumptions, or are 
downright impossible. Recent works~\cite{C:BCKM21b,EC:GLSV21} have shown that there exist quantum protocols for oblivious
transfer, and therefore arbitrary multi-party computation (MPC), based solely on the
existence of one-way functions (OWF)~\cite{C:BCKM21b,EC:GLSV21}, or pseudorandom states (PRS)~\cite{C:JiLiuSon18},
which potentially entail even weaker computational assumptions~\cite{Kre21,kretschmer2022quantum}.
It is well-known that, classically, oblivious transfer and MPC are ``Cryptomania'' objects, i.e.,
they can only be constructed from more structured assumptions that imply public-key encryption
(PKE). Thus, the above results seem to challenge the boundary between Cryptomania and MiniCrypt,
in the presence of quantum information. Motivated by this state of affairs, in this work we
investigate the notion of \emph{PKE itself}, the heart of Cryptomania, through the lenses of quantum computing. That is, we ask the following question:
\begin{quote}\centering
\emph{Does public-key encryption (PKE) belong to MiniQCrypt?}
\end{quote}
Known results around this question are mostly negative: It is known that PKE cannot be constructed in a black-box manner from OWFs~\cite{C:ImpRud88}, and this result has been recently re-proven in the more challenging setting where 
the encryption or decryption algorithms are quantum~\cite{C:ACCFLM22}. However, a tantalizing possibility left open by these works is to realize PKE schemes from OWFs (or weaker assumptions), where public-key or ciphertexts are quantum states.

\subsection{Our results}

In this work we initiate the systematic study of quantum public-key encryption (qPKE), i.e., public-key encryption where public-keys and ciphertexts are allowed to be quantum states. We break down our contributions as follows.

\paragraph{1.~Definitions.}
We provide a general definitional framwork for qPKE, where both the public-key and
ciphertext might be general quantum states.
In the classical setting, there is no need to provide oracle access to the encryption, since the public-key can be used to implement that. In contrast, if the public-key is a quantum state, it might be
measured during the encryption procedure, and the ciphertexts might depend on the measurement outcome. In fact, this is the approach taken in some of our constructions. This motivates a stronger security definition, similar to the classical counterpart, in which the adversary gets additional access to an encryption oracle that uses the same quantum public-key that is used during the challenge phase. 
We define IND-CPA-EO (respectively, IND-CCA-EO) security by adding the encryption oracle (EO) to the standard IND-CPA (respectively, IND-CCA) security game.

\paragraph{2.~Constructions.}
With our new security definition at hand, we propose three protocols for
implementing qPKE from OWF and potentially weaker assumptions, each with its own different advantages and
disadvantages.
More concretely, we show the existence of:
\begin{enumerate}
  \item 
    A qPKE scheme
    with quantum public-keys and classical ciphertexts that is IND-CCA-EO\footnote{Throughout this paper, unless explicitly specified, by IND-CCA we refer to the notion of adaptive IND-CCA2 security.} secure, based on post-quantum OWF, in \cref{sec:cca_qpke}.

  \item
    A qPKE scheme with quantum public-key and quantum ciphertext that is IND-CCA1 secure, based on pseudo-random function-like states (PRFS) with super-logarithmic input-size\footnote{Note that PRS implies PRFS with
    logarithmic size inputs, but no such implication is known for super-logarithmic inputs.}, in \cref{sect:cca_from_prfs}. Since this scheme is not EO secure, each quantum public-key enables the encryption of a single message.
      \item %
    A qPKE scheme with quantum public-key and
    classical ciphertext that is IND-CPA-EO secure based on pseudo-random function-like states with proof of destruction (PRFSPDs), in \cref{sect:qpke-from-prfspd}. 
\end{enumerate}
We wish to remark that it has been recently shown that OWF imply PRFS with
super-logarithmic input-size~\cite{C:AnaQiaYue22} and PRFSPDs~\cite{cryptoeprint:2023/543}.
Therefore, the security of the second and third protocols is based on a
potentially weaker cryptographic assumption than the first one.
Furthermore, PRFS with super-logarithmic input-size are \emph{oracle separated} from
one-way functions~\cite{Kre21}; therefore, our second result shows a black-box
separation between a certain form of quantum public-key encryption and one-way
functions.
On the other hand, for the other two constructions, even if
the public-key is a quantum state, the ciphertexts are classical and, furthermore, one quantum
public-key can be used to encrypt multiple messages.
The first protocol is much simpler to describe and understand since
it only uses standard (classical) cryptographic objects.
Moreover, we show that this scheme guarantees the notion of adaptive CCA2
security and is the only scheme that achieves perfect correctness.

\paragraph{3.~Lower Bounds.} To complete the picture, we demonstrate that \emph{information-theoretically secure} qPKE does not exist. Due to the public-keys being quantum states, this implication is much less obvious than for the classical case. In fact, some of the existing constructions of qPKE~\cite{GottesmanConstruction} have been conjectured to be unconditionally secure, a conjecture that we invalidate in this work. While this general statement follows by known implications in the literature (see~\cref{sec:impossibility} for more details), in this work we present a self-contained proof of this fact, borrowing techniques from shadow tomography, which we consider to be of independent interest.

\subsection{Technical overview}
\label{sec:technical_overview}
In this section, we provide a technical overview of our results. In  \Cref{sec:intro-definition}, we explain the challenges and choices to define qPKE and its security definition. In \Cref{sec:intro-constructions}, we present 3 instantiations of qPKE, each based on a different assumption and with different security guarantees. Ultimately, \Cref{sec:impossibility_overview} is dedicated to the impossibility of information-theoretically secure qPKE and a high-level overview of the proof technique. 

\subsubsection{Definitions of qPKE}
\label{sec:intro-definition}
In order to consider public-key encryption schemes with quantum public-keys, we
need to revisit the traditional security definitions%
.
In the case of quantum public-keys, there are several immediate issues that require revision. 

The first issue is related to the access the adversary is given to the public-key. In the classical-key case (even with quantum ciphertexts), the adversary is
given the classical public-key $\pk$. Given a single quantum public-key,  one cannot create arbitrary number of copies of the quantum public-key, due to no-cloning. Hence, to naturally extend notions such as IND-CPA security, we provide multiple copies of the quantum public-key to the adversary (via the mean of oracle access to the quantum public-key generation algorithm).

The second issue concerns the quantum public-key's \emph{reusability}. Classically, one can use the public-key to encrypt multiple messages. With quantum public-keys, this might not be the case: the quantum public-key might be consumed during the encryption. 
In a non-reusable scheme, the user needs a fresh quantum public-key for every plaintext they wish to encrypt. This is not only a theoretical concern: in the PRFS-based construction (see \cref{sect:cca_from_prfs}), part of the quantum public-key is sent as the (quantum) ciphertext, so clearly, this construction is \emph{not} reusable. %

Thirdly, it could be the case that in a reusable scheme, each encryption call changes the public-key state $\rho_\qpk$ in an irreversible way. %
Hence, we make a syntactic change: $\enc(\rho_\qpk, m)$ outputs
$(\rho_\qpk',c)$, where $c$ is used as the ciphertext and $\rho_\qpk'$ is used as the
key to encrypt the next message. Note that in this scenario the updated public-key is not publicly available anymore and is only held by the party who performed the encryption.

Lastly, the syntactic change mentioned above also has security effects.
Recall that classically, there is no need to give the adversary access to an encryption oracle, since the adversary can generate encryption on their own. Alas, with quantum public-keys, the distribution of ciphers might depend on the changes that were made to the quantum public-key by the challenger whenever the key is used to encrypt several messages. Therefore, for reusable schemes, we define two new security notions, denoted CPA-EO and CCA-EO, that are similar to CPA and CCA but where the adversary is given access to an encryption oracle (EO). 
We note there are several works considering the notions of chosen-ciphertext security in the quantum setting, because it is not clear how to prevent the adversary from querying the challenge ciphertext, if it contains a quantum states. However, we only consider CCA-security for schemes with classical ciphertexts, and therefore this issue does not appear in this work.

\paragraph{Pure vs Mixed States.} We mention explicitly that we require our public-keys to be \emph{pure states}. This is motivated by the following concern: there is no general method to authenticate quantum states. One proposal to ensure that the certificate authority (CA) is sending the correct state is to distribute various copies of the keys to different CAs and test whether they are all sending the same state~\cite{GottesmanConstruction}. This ensures that, as long as at least one CA is honest, the user will reject a malformed key with some constant probability. However, this argument crucially relies on the public-key being a pure state (in which case comparison can be implemented with a SWAP-test). On the other hand, if the public-key was a mixed state, there would be no way to run the above test without false positives.

We also mention that, if mixed states are allowed, then there is a trivial
construction of qPKE from any given symmetric encryption scheme
$(\mathsf{SKE}.\keygen, \allowbreak \mathsf{SKE}.\enc,\mathsf{SKE}.\dec)$, as
also observed in~\cite[Theorem C.6]{cryptoeprint:2022/1336}, which we describe in the
following.
To generate the keys, we use the output of $\mathsf{SKE}.\keygen$ as the
secret-key and use it to create the uniform mixture
\begin{equation}
\frac{1}{2^n} \sum_{x\in\{0,1\}^n}\ketbra{x}\tensor \ketbra{\enc_{sk}(x)}    
\label{eq:pk_for_encryption_with_mixture}
\end{equation}
as the public-key.
The ciphertext corresponding to a message $m$ is given by $(\enc_x(m),\enc_{sk}(x))$.
To decrypt, the decryptor would first recover $x$ by decrypting the second
element in the ciphertext using $sk$, and then recover $m$ by decrypting the
first item using $x$ as the secret key.

\subsubsection{Constructions for qPKE}
\label{sec:intro-constructions}
As previously mentioned, we propose in this work three schemes for qPKE, based on three different assumptions, each providing a different security guarantee. 

\paragraph{qPKE from OWF.} 
Our first instantiation of qPKE is based on the existence of post-quantum OWFs.
For this construction, we aim for the strong security notion of indistinguishability against adaptive chosen ciphertext attacks with encryption oracle referred to as IND-CCA-EO. We start with a simple bit-encryption construction that provides IND-CCA security and we discuss how one can modify the scheme to encrypt multi-bit messages and also provide EO security. 

Our first scheme assumes the existence of a \emph{quantum-secure pseudorandom
  function~(PRF)}, which can be built from quantum-secure one-way
functions~\cite{FOCS:Zhandry12}.
Given a PRF ensemble $\{f_{k}\}_k$, the public key consists of a pair of pure
quantum states $\qpk=(\ket{\qpk_0}, \ket{\qpk_1})$ and the secret key consists
of a pair of bit-strings $\dk=(\dk_0,\dk_1)$ such that, for all~$b\in\{0,1\}$,
\[ \ket{\qpk_b} = \frac1{\sqrt{2^n}} \sum_{x\in\{0,1\}^n} \ket{x,f_{\dk_b}(x)}, \]
where $f_k$ denotes the quantum-secure PRF keyed by~$k$.
To encrypt a bit~$b$, one simply measures all qubits of $\ket{\qpk_b}$ in the
computational basis.
The result takes the form $(x,f_{\dk_b}(x))$ for some uniformly
random~$x\in\{0,1\}^n$ and this is returned as the ciphertext, i.e.,
$(\qc_0,\qc_1) = (x,f_{\dk_b}(x))$.

To decrypt a ciphertext $(\qc_0,\qc_1)$, we apply both~$f_{\dk_0}$ and~$f_{\dk_1}$
to~$\qc_0$ and return the value of~$b\in\{0,1\}$ such
that~$f_{\dk_b}(\qc_0) = \qc_1$.
In case this does happen for neither or both of the keys, the decryption aborts. %

The IND-CCA security of the simple bit-encryption scheme can be proven with a hybrid argument~(see~\cref{app:CCA_bit_encryption}). However, there are a few caveats to the scheme that can be pointed out. First, the scheme is not reusable. It can be easily noticed that after using a public-key for an encryption, the public-key state collapses, meaning that all the subsequent encryption calls are derandomized. This would mean if the same public-key is reused, it can not even guarantee IND-CPA security as the encryption is deterministic.

The second issue is lifting this CCA-secure bit-encryption scheme to a many-bit CCA-secure encryption scheme. Note that although not trivial, as proven by Myers and Shelat~\cite{FOCS:MyeShe09}, classically it is possible to construct CCA-secure many-bit encryption from CCA-secure bit-encryption. However, the argument cannot be extended to qPKE in a generic way.
The main issue is that the construction from~\cite{FOCS:MyeShe09}, similar to
the Fujisaki-Okamoto transform, derandomizes the encryption procedure for some fixed random coins. %
Later these fixed random coins are encrypted and attached to the ciphertext, so that
the decryptor can re-encrypt the plaintext to make sure they were handed the
correct randomness.
Looking at our construction, it is quite clear that it is not possible to
derandomize the encryption procedure as the randomness is a consequence of the
measurement.

Let us show how the same approach can be modified to circumvent the issues mentioned. Our main observation is that we can use public-keys of the form mentioned before for a key agreement stage and then use the agreed key to encrypt many-bit messages with a symmetric-key encryption scheme (SKE). Let us elaborate. 
Let $\{f_{k}\}_k$ be a PRF family and $(\mathsf{SE}.\enc,\mathsf{SE}.\dec)$ be a symmetric-key encryption scheme. Note that quantum-secure one-way functions imply a quantum-secure PRF~\cite{FOCS:Zhandry12}, and post-quantum IND-CCA symmetric encryption~\cite{EC:BonZha13}\footnote{IND-CCA SKE can be built from an IND-CPA SKE and a MAC using the encrypt-then-MAC paradigm.}. %
Consider the following scheme: the secret key $\dk$ is a uniformly random key for the PRF, and for a fixed $\dk$,
the quantum public-key state is
\begin{align}\label{eq:public-key-owf}
\ket{\qpk_\dk}=\frac{1}{\sqrt{2^\secpar}} \sum_{x \in
  \{0,1\}^\secpar}\ket{x}\ket{f_\dk(x)}.
\end{align}
The encryption algorithm will then measure $\ket{\qpk_\dk}$ in the computational
basis leading to the outcome $(x^*, y^*=f_{\dk}(x^*))$. The ciphertext of a message
$m$ is given by $(x^*,\mathsf{SE}.\enc(y^*,m))$.
To decrypt a ciphertext $(\hat{x},\hat{c})$, we first compute $\hat{y} = f_{\dk}(x)$ and return $\hat{m} = \mathsf{SE}.\dec(f_{\dk}(\hat{x}),\hat{c})$.

We emphasize that this scheme is reusable since it allows the encryption %
 of many messages using the same
measurement outcome $(x^*, f_{\dk}(x^*))$.
Using a hybrid argument, it can be shown that if the underlying SKE guarantees IND-CCA security, this construction fulfills our strongest security notion, i.e. IND-CCA-EO security. A formal description of the scheme, along with a security proof can be found in \cref{sec:cca_qpke}.

\paragraph{QPKE from PRFS.} The second construction we present in this paper is an IND-CCA1 secure public-key scheme based on the existence of pseudorandom function-like state generators. Our approach is based on first showing bit-encryption, and the discussion regarding how to lift that restriction is discussed in \cref{sect:cca_from_prfs}. The ciphertexts generated by our scheme are quantum states, and as the public-keys of this construction are not reusable, we do not consider the notion of EO security. A family of states $\{\ket{\psi_{k,x}}\}_{k,x}$ is pseudo-random
function-like~\cite{C:AnaQiaYue22} if
\begin{enumerate}
\item There is a quantum polynomial-time algorithm $\gen$ such that
\[\gen(k,\sum_{x} \alpha_x \ket{x}) = \sum_{x} \alpha_x \ket{x}\ket{\psi_{k,x}} \text{, and}\]
  \item No \(\qpt\) adversary can distinguish $(\ket{\psi_1},...,\ket{\psi_\ell})$ from
    $(\ket{\phi_1},...,\ket{\phi_\ell})$, where \allowbreak
    $\ket{\psi_i} = \sum_{x} \alpha^i_x \ket{x}\ket{\psi_{k,x}}$,
    $\ket{\phi_i} = \sum_{x} \alpha^i_x \ket{x}\ket{\phi_{x}}$, and $\{\ket{\phi_x}\}_x$ are Haar
    random states and the states $\ket{\sigma_i} = \sum_{x} \alpha^i_x \ket{x}$ are chosen by
    the adversary.
\end{enumerate}
We continue by providing a high-level description of the scheme. The key generation algorithm picks a uniform PRFS key $\dk$ and generates the corresponding public-keys as stated below: 
\begin{align}\label{eq:public-key-prfs}
  \frac{1}{\sqrt{2^\secpar}} \sum_{x \in \{0,1\}^\secpar}\ket{x}\ket{\psi_{\dk,x}}^{\ot n},
\end{align}
where $\{\ket{\psi_{k,x}}\}_{k,x}$ is a PRFS family, the size of the input $x$ is super-logarithmic in the security parameter and $n$ is a polynomial in the security parameter. 

To encrypt a bit $m$, the encryptor will then measure the first register of $\ket{\qpk}$ to obtain $x^*$ and the 
residual state after this measurement will be of form $\ket{x^*}\ket{\psi_{\dk,x^*}}^{\ot n}$. They also sample a uniform key $\dk_1$ and compute the state $\ket{\psi_{\dk_1,x^*}}$ %
then compute the ciphertext $c = (x^*,\rho)$ where
\begin{align}
\rho =  \begin{cases} \ket{\psi_{\dk,x^*}}^{\ot n}, & \text{if } m = 0 \\ \ket{\psi_{\dk_1,x^*}}^{\ot n},& \text{if } m = 1\end{cases}.
\end{align}
To decrypt a ciphertext $(\hat{x}, \hat{\rho})$, we first compute $n$ copies of the state $\ket{\psi_{\dk ,\hat{x}}}$ and performs a SWAP tests between each copy and the subsystems of $\hat{\rho}$ with the same size as $\ket{\psi_{\dk,\hat{x}}}$. If all the SWAP tests return $0$ the decryption algorithm returns $\hat{m} = 0$ and otherwise it returns $\hat{m} = 1$.  For a large enough $n$, our scheme achieves statistical correctness.

We prove that this construction guarantees IND-CCA1 security by a hybrid argument in \cref{sect:cca_from_prfs}. We emphasize that as the ciphertexts of the scheme are quantum states it is challenging to define adaptive CCA2 security.

\paragraph{QPKE from PRFSPDs.}
Our third scheme is based on pseudo-random function-like states with proof of
destruction (PRFSPDs), which was recently defined
in~\cite{cryptoeprint:2023/543}.
The authors extended the notion of PRFS to pseudo-random
function-like states with proof of destruction, where we have two algorithms
$\destruct$ and $\ver$, which allows us to verify if a copy of the PRFS was
destructed.

We will discuss now how to provide non-reusable CPA security security\footnote{Meaning that
one can only encrypt once using a $\ket{\qpk}$.} of the encryption of a one-bit
message and we discuss later how to use it to achieve reusable security, i.e., CPA-EO security.
The quantum public-key in this simplified case is
\begin{align}\label{eq:public-key-prfspd}
  \frac{1}{\sqrt{2^\secpar}} \sum_{x \in \{0,1\}^\secpar}\ket{x}\ket{\psi_{\dk,x}}.
\end{align}
The encryptor will then measure the first register of $\ket{\qpk}$ and the
post-measurement state is $\ket{x^*}\ket{\psi_{\dk,x^*}}$. The encryptor will
then generate a (classical) proof of destruction $\pi = \destruct(\ket{\psi_{\dk,x^*}})$. The
encryption procedure also picks $\dk_1$ uniformly at random, generated $\ket{\psi_{\dk_1,x^*}}$ and generates the proof of destruction $\pi'=\destruct(\ket{\psi_{\dk_1,x^*}})$. 
The corresponding ciphertext for a bit $b$ is given by $c = (x^*,y)$, where \[y =  \begin{cases}\pi', & \text{if } b = 0 \\ \pi,& \text{if
} b = 1\end{cases}.\]
The decryptor will receive some value $(\hat{x},\hat{y})$ and
decrypt the message $\hat{b} =  \ver(\dk,\hat{x}, \hat{y})$.
The proof of the security of the aforementioned construction follows from a hybrid argument reminiscent of the security proof of the previous schemes~(see \cref{sect:qpke-from-prfspd}).
Notice that repeating such a process in parallel trivially gives a one-shot
security of the encryption of a string $m$ and moreover, such an encryption is
classical. Therefore, in order to achieve IND-CPA-EO secure qPKE scheme, we can actually
encrypt a secret key $\sk$ that is chosen by the encryptor, and send the message
encrypted under $\sk$. We leave the details of such a construction and its proof
of security to~\cref{sect:qpke-from-prfspd}.
 
\subsubsection{Impossibility of Information-Theoretically Secure
  qPKE}\label{sec:impossibility_overview}
So far, we have established that qPKE can be built from assumptions weaker than the ones required for the classical counterpart, and potentially even weaker than those needed to build secret-key encryption classically. %
This naturally leads to the question of whether it is possible to build an information-theoretically secure qPKE. In the following, we present a self-contained proof of this fact, using techniques from the literature on shadow tomography.
Although proving the impossibility for classical PKE is immediate, there are a few challenges when trying to prove a result of a similar flavor for qPKE. Even when considering security against a computationally unbounded adversary, there is a limitation that such adversary has, namely, they are only provided with polynomially many copies of the public-key.

The first step of the proof is reducing winning the IND-CPA game to finding a secret-key/public-key pair $(\dk,\ket{\qpk_{\dk}})$ such that 
\[
\langle\qpk^* \ket{\qpk_{\dk}} \approx 1.
\]
In other words, we show that if $\ket{\qpk_{\dk}}$ is relatively close to $\ket{\qpk^*}$, there is a good chance that $\dk$ can decrypt ciphertexts encrypted by $\ket{\qpk^*}$ correctly. A formal statement and the proof of this argument can be found in \cref{lemma:pk_distance}. 

Given this lemma, the second part of the proof consists in constructing an adversary that takes polynomially many copies of $\ket{\qpk^*}$ as input and outputs $(\dk,\ket{\qpk_{\dk}})$ such that $\ket{\qpk_{\dk}}$ is relatively close to $\ket{\qpk^*}$. The technique to realize this adversary is \emph{shadow tomography}, which shows procedures to estimate 
the values $\langle \qpk_{\dk} \vert \qpk^*\rangle$ for all $(\ket{\qpk_{\dk}} , \dk)$ pairs. Note that doing this naively, i.e. by SWAP-testing multiple copies of $\ket{\qpk^*}$ with each $\ket{\qpk_{\dk}}$, would require exponentially many copies of the public-key $\ket{\qpk^*}$. The way we circumvent this problem is by using the a recent result by Huang, Kueng, and Preskill~\cite{huang2020predicting}. Informally, this theorem states that for $M$ rank 1 projective measurements $O_1,\dots,O_M$ and an unknown $n$-qubit state $\rho$, it is possible to estimate $\trace(O_i\rho)$ for all $i$, up to precision $\epsilon$, by only performing $T = O(\log(M)/\epsilon^2)$ single-copy random Clifford measurements on $\rho$. 

Employing this theorem, we show that a computationally unbounded adversary can estimate all the values $\langle \qpk_{\dk} \vert \qpk^*\rangle$ from random Clifford measurements on polynomially many copies of $\ket{\qpk^*}$. Having the estimated values of $\langle \qpk_{\dk} \vert \qpk^*\rangle$ the adversary picks a $\dk$ such that the estimated value is relatively large and uses this key to decrypt the challenge ciphertext. Now invoking Lemma~\ref{lemma:pk_distance} we conclude that the probability of this adversary winning the IND-CPA game is significantly more than $1/2$.

\subsection{Related works} %
\label{sub:related_works}
The notion of qPKE was already considered in the literature, although without introducing formal security definitions.
For instance, Gottesman~\cite{GottesmanConstruction} proposed a candidate construction in an oral presentation, without a formal security analysis. The scheme has quantum public-keys and quantum
ciphers, which consumes the public-key for encryption. Kawachi et al.~\cite{EC:KKNY05} proposed a construction of qPKE (with quantum keys and ciphertexts) from a newly introduced hardness assumption, related to the graph automorphism problem. \cite{C:OkaTanUch00} defines and constructs a public-key encryption where the keys,
plaintexts and ciphers are classical, but the algorithms are quantum the (key-generation uses Shor's algorithm).
One of the contributions of this work, is to provide a unifying framework for these results, as well as improve in terms of computational assumptions and security guarantees.

In~\cite{NI09}, the authors define and provide impossibility results regarding
encryption with quantum public-keys.
Classically, it is easy to show that a (public) encryption scheme cannot have
deterministic ciphers; in other words, encryption must use randomness.
They show that this is also true for a quantum encryption scheme with quantum
public-keys.
In~\cite{cryptoeprint:2020/1557}, a secure encryption scheme with quantum public
keys based on the LWE assumption is introduced.
That work shows (passive) indistinguishable security, and is not IND-CPA secure.

In~\cite{C:MorYam22,cryptoeprint:2022/1336}, the authors study digital
signatures with quantum signatures, and more importantly in the context of this
work, quantum public-keys.

The study of quantum pseudorandomness and its applications has recently experienced rapid advancements. One of the most astonishing aspects is that PRS (Pseudorandom states) and some of its variations are considered weaker than one-way functions. In other words, they are implied by one-way functions, and there exists a black-box separation between them. However, it has been demonstrated that these primitives are sufficient for many applications in Minicrypt and even extend beyond it.  A graph presenting the various notions of quantum pseudorandomness and its application is available  at~\url{https://sattath.github.io/qcrypto-graph/}.

\subsection{Concurrent and subsequent work}
\ifanon
\else
    This work is a merge of two concurrent and independent works~\cite{cryptoeprint:2023/306,cryptoeprint:2023/345}, with a unified presentation and more results.
    
\fi
In a concurrent and independent work, Coladangelo~\cite{cryptoeprint:2023/282}
proposes a qPKE scheme with a construction that is very different from ours, and
uses a quantum trapdoor function, which is a new notion first introduced in
their work.
Their construction is based on the existence of quantum-secure OWF.
However, in their construction, each quantum public-key can be used to encrypt a
single message (compared to our construction from OWF, where the public-key can
be used to encrypt multiple messages), and the ciphertexts are quantum (whereas
our construction from OWF has classical ciphertexts).
They do not consider the stronger notion of IND-CCA security.

Our paper has already generated interest in the community: Two follow-up
works~\cite{cryptoeprint:2023/490,cryptoeprint:2023/500} consider a \emph{stronger} notion of qPKE
where the public-key consists of a classical and a quantum part, and the
adversary is allowed to tamper arbitrarily with the quantum part (but not with
the classical component).\footnote{Because of this stronger security definition,
  here the notion of public-keys with mixed states is meaningful since there is
  an alternative procedure to ensure that the key is well-formed (e.g., signing
  the classical component).}
The authors provide constructions assuming quantum-secure OWF.
While their security definition is stronger, we remark that our approach is more
general, as exemplified by the fact that we propose constructions from
potentially weaker computational assumptions.
In~\cite{barhoush2023sign}, the authors give another solution for the quantum public-key distribution problem using time-dependent signatures, which can be constructed from quantum-secure OWF, but the (classical) verification key needs to be continually updated.

\section{Preliminaries}
\label{sec:definitions}
\subsection{Notation}
Throughout this paper, \(\secpar\) denotes the security parameter.
The notation \(\negl\) denotes any function \(f\) such that \(f(\lambda) = \lambda^{-\omega(1)}\), and
\(\poly\) denotes any function \(f\) such that \(f(\lambda) = \lambda^{\mathcal{O}(1)}\).
When sampling uniformly at random a value \(a\) from a set~\(\mathcal{U}\), we
employ the notation \(a \sample \mathcal{U}\).
When sampling a value \(a\) from a probabilistic algorithm \(\adv\), we employ the
notation \(a \leftarrow \adv\).
Let \(\size{\cdot}\) denote either the length of a string, or the cardinal of a finite
set, or the absolute value.
By \(\ppt\) we mean a polynomial-time non-uniform family of probabilistic
circuits, and by \(\qpt\) we mean a polynomial-time family of quantum circuits.

\subsection{Quantum Information}
For a more in-depth introduction to quantum information, we refer the reader to
\cite{DBLP:books/daglib/0046438}.
We denote by \(\Hilb_{M}\) a complex Hilbert space with label \(M\) and finite
dimension \(\dim{M}\) .
We use the standard bra-ket notation to work with pure states
\(\ket{\psi} \in \Hilb_{M}\).
The class of positive, Hermitian, trace-one linear operators on \(\Hilb_{M}\) is
denoted by \(\mathcal{D}(\Hilb_{M})\).
A quantum register is a physical system whose set of valid states is
\(\mathcal{D}(\Hilb_{M})\); in this case we label by \(M\) the register itself.
The maximally mixed state (i.e., uniform classical distribution) is written as
\(\ident / \dim M\) on \(M\).

The support of a quantum state \(\varrho\) is its cokernel (as a linear operator).
Equivalently, this is the span of the pure states making up any decomposition of
\(\varrho\) as a convex combination of pure states.
We will denote the orthogonal projection operator onto this subspace by
\(P^{\varrho}\).
The two-outcome projective measurement (to test if a state has the same or
different support as \(\varrho\)) is then \(\{P^{\varrho}, \ident - P^{\varrho}\}\).

A quantum \emph{\(t\)-design} (for a fixed \(t\)) is a probability distribution
over pure quantum states which can duplicate properties of the probability
distribution over the Haar measure for polynomials of degree \(t\) or less.
A quantum \(t\)-design with \(n\)-qubit output can be efficiently implemented
with a random \(\poly[t, n]\)-size quantum circuits.

We recall the SWAP test on two quantum states \(\ket{\psi}, \ket{\phi}\) which is an
efficient algorithm that outputs \(0\) with probability
\(\frac{1}{2} + \frac{1}{2}\abs{\braket{\psi}{\phi}}^{2}\).
In particular, if the states are equal, the output of SWAP test is always \(0\).

Next, we state a well-known fact about the quantum evaluation of classical
circuits.

\begin{fact}\label{thm:quantum}
  Let $f\colon\{0,1\}^n \to \{0,1\}^m$ be a function which is efficiently
  computable by a classical circuit.
  Then there exists a unitary~$U_f$ on $(\CC^2)^{\ot n+m}$ which is efficiently
  computable by a quantum circuit (possibly using ancillas) such that, for
  all~$x \in\{0,1\}^n$ and~$y\in\{0,1\}^m$,
\begin{align*}
    U_f\colon \ket{x}\ket{y} \mapsto \ket{x}\ket{y \oplus f(x)}.
\end{align*}
\end{fact}

\subsection{Quantum-Secure Pseudorandom Functions}\label{subsec:prf}
Throughout this work, we often refer to a \emph{pseudorandom function} (PRF)
first introduced in~\cite{goldreich1986construct}.
This is a keyed function, denoted~$\PRF$, that can be evaluated in polynomial
time satisfying a certain security property.
In this work, we require $\PRF$ to be \emph{quantum-secure}, which, loosely
speaking, says that an adversary with oracle access to $\PRF$ cannot distinguish
it from a truly random function, even given superposition queries.
It is known that quantum-secure PRF can be constructed from any quantum-secure
one-way function~\cite{FOCS:Zhandry12}.

\begin{definition}[Quantum-secure PRF]
  We say that a keyed family of functions $\{f_k\}_k$ is a \emph{quantum-secure
    pseudorandom function (PRF) ensemble} if, for any \(\qpt\) adversary~$\adv$,
  we have
\begin{align*}
  \abs{
    \Pr\left[ 1 \gets \adv(1^\lambda)^{f_k} \right]
  - \Pr\left[ 1 \gets \adv(1^\lambda)^{f} \right]
  }
  \leq \mu(\lambda),
\end{align*}
where $k\xleftarrow{\$}\{0,1\}^\lambda$, $f$ is a truly random function, and the oracles can be accessed in superposition, that is, they implement the following unitaries
\[
  \ket{x}\ket{z} \xmapsto{U_{f_k}} \ket{x}\ket{z\oplus f_k(x)}
\quad\text{and}\quad
  \ket{x}\ket{z} \xmapsto{U_f} \ket{x}\ket{z\oplus f(x)},
\]
respectively.
\end{definition}

\subsection{Post-Quantum IND-CCA Symmetric-Key Encryption}
We briefly recall the definition of a symmetric-key encryption scheme (SKE).
\begin{definition}
An SKE consists of 2 algorithms with the
following syntax:
\begin{enumerate}
  \item $\enc(\sk, \pt)$: a $\ppt$ algorithm, which receives a symmetric-key $\sk\in\{0,1\}^\secpar$ and a plaintext $\pt$, and outputs a ciphertext $\ct$.
  \item $\dec(\sk, \ct)$: a deterministic polynomial-time algorithm, which takes a symmetric-key \(\sk\) and a ciphertext \(\ct\), and outputs a plaintext
    \(\pt
    \).
\end{enumerate}
\end{definition}

We say that a SKE scheme is perfectly \emph{correct} if for every plaintext
\(\pt \in \{0,1\}^*\) and symmetric-key $\sk\in \{0,1\}^\lambda$, $\dec(\sk,\enc(\sk,\pt)) = \pt$. 

\begin{definition}
  \label{def:cca-ske}
  An SKE is post-quantum IND-CCA secure if for every \(\qpt\)
  adversary \(\adv \coloneqq (\adv_{1}, \adv_{2})\), there exists a negligible
  function \(\epsilon\) such that the following holds for all $\lambda$:
\begin{align*}
\Pr\left[
\Tilde{b} =b 
\ \middle\vert
\begin{array}{ll}
\sk \xleftarrow{\$} \{0,1\}^{\lambda}  \\
\pt_0,\pt_1 \gets \adv_1^{\enc(\sk,\cdot),\dec(\sk,\cdot)}(1^\lambda)  \\
b\xleftarrow{\$} \{0,1\} \\ 
\ct^* \gets \enc(\sk,\pt_b)\\  
\tilde{b} \gets \adv_2^{\enc(\sk,\cdot) ,\dec^*(\sk,\cdot)}(\ct^*,1^\lambda)
\end{array}
\right] \leq 1/2 + \epsilon(\lambda),
\end{align*}
Where $\dec^*(\sk,\cdot)$ is the same as $\dec(\sk,\cdot)$ but returns $\bot$ on input the challenge ciphertext $\ct^*$.

\end{definition}
Note that as the adversary is not given superposition access to the $\enc$, $\dec$ oracles one can build post-quantum IND-CCA SKE from quantum-secure OWF the same way as it is done classically with message authentication codes. 

\subsection{Pseudorandom Function-Like State (PRFS) Generators}

The notion of pseudorandom function like states was first introduced by Ananth,
Qian and Yuen in~\cite{C:AnaQiaYue22}.
A stronger definition where the adversary is allowed to make superposition
queries to the challenge oracles was introduced in the follow-up
work~\cite{TCC:AGQY22}.
We reproduce their definition here:

\begin{definition}[Quantum-accessible PRFS generator]
  \label{def:prfs}
We say that a $\qpt$ algorithm $G$ is a quantum-accessible secure pseudorandom function-like state generator if for all $\qpt$ (non-uniform) distinguishers $A$ if there exists a negligible function $\epsilon$, such that for all $\lambda$, the following holds:
 \[
    \left|\Pr_{k \leftarrow \{0,1\}^{\secparam} }\left[A_\lambda^{\ket{{\cal O}_{\sf PRFS}(k,\cdot)}}(\rho_\lambda) = 1\right] - \Pr_{{\cal O}_{\sf Haar}}\left[A_\lambda^{\ket{{\cal O}_{\sf Haar}(\cdot)}}(\rho_\lambda) = 1\right]\right| \le \epsilon(\lambda),
  \]
  where:
\begin{itemize}
    \item ${\cal O}_{\sf PRFS}(k,\cdot)$, %
    on input a $d$-qubit register ${\bf X}$, does the following: it  applies an isometry channel that is controlled on the register ${\bf X}$ containing $x$, it creates and stores  $G_{\secparam}(k,x)$ in a new register ${\bf Y}$.  It outputs the state on the registers ${\bf X}$ and ${\bf Y}$.
    \item ${\cal O}_{\sf Haar}(\cdot)$, modeled as a  channel, on input a $d$-qubit register ${\bf X}$, does the following: it applies a channel that controlled on the register ${\bf X}$ containing $x$, stores $\ketbra{\vartheta_x}$ in a new register ${\bf Y}$, where $\ket{\vartheta_x}$ is sampled from the Haar distribution. It outputs the state on the registers ${\bf X}$ and ${\bf Y}$.
\end{itemize}
Moreover, $A_{\secparam}$ has superposition access to ${\cal O}_{\sf PRFS}(k,\cdot)$ and ${\cal O}_{\sf Haar}(\cdot)$ (denoted using the ket notation).
\par We say that $G$ is a $(d(\lambda),n(\lambda))$-QAPRFS generator to succinctly indicate that its input length is $d(\lambda)$ and its output length is $n(\lambda)$.
\end{definition}

\subsection{Quantum Pseudorandomness with Proofs of Destruction}
\label{sect:def-prfspd}
We import the definition of pseudorandom function-like states with proofs of destruction (PRFSPD) from \cite{cryptoeprint:2023/543}.

\begin{definition}[PRFS generator with proof of destruction]%
  A $\PRFSPD$ scheme with key-length $w(\secpar)$, input-length $d(\secpar)$, output length $n(\secpar)$ and proof length $c(\secpar)$ is a tuple of  $\qpt$ algorithms $\gen,\destruct,\ver$ with the following syntax:
  \begin{enumerate}
       \item $\ket{\psi^x_k}\gets \gen(k,x)$: takes a key $k\in \{0,1\}^w$, an input string $x\in \{0,1\}^{d(\secpar)}$, and outputs an $n$-qubit pure state
        $\ket{\psi^x_k}$.
      \item $p\gets \destruct(\ket{\phi})$: takes an $n$-qubit quantum state $\ket{\phi}$ as input, and outputs a $c$-bit classical string, $p$.
      \item $b\gets \ver(k,x,q)$: takes a key $k\in \{ 0,1 \}^w$, a $d$-bit input string $x$, a $c$-bit classical string $p$ and outputs a boolean output $b$.
  \end{enumerate}

  \paragraph{Correctness.} A $\PRFSPD$ scheme is said to be correct if for every $x\in \{0,1\}^{d}$,
  \[\Pr_{k\uniform{w}}[1\gets\ver(k,x,p)\mid p\gets \destruct(\ket{\psi^x_k}); \ket{\psi^x_k}\gets \gen(k,x)]=1\]
  \paragraph{Security.}
  \begin{enumerate}
      \item \textbf{Pseudorandomness:} A $\PRFSPD$ scheme is said to be (adaptively) pseudorandom if for any $\qpt$ adversary $\adv$, and any polynomial $m(\secpar)$, there exists a negligible function $\negl$, such that
      \begin{align*}&|\Pr_{ k\gets\{0,1\}^w}[\adv^{\ket{\gen(k,\cdot)}}(1^\secpar)=1]\\
      &-\Pr_{ \forall x\in \{0,1\}^d, \ket{\phi^x}\gets\mu_{(\CC^2)^{\tensor n}}}[\adv^{\ket{\haar^{\{\ket{\phi^x}\}_{x\in \{0,1\}^d}}(\cdot)}}(1^\secpar)=1]|=\negl
      \end{align*}
      where $\forall x\in \{0,1\}^d$, $\haar^{\{\ket{\phi^x}\}_{x\in \{0,1\}^d}}(x)$ outputs $\ket{\phi^x}$. Here we note that $\adv$ gets quantum access to the oracles.\label{item:pseudorandomness}
      \item \textbf{Unclonability{-}of{-}proofs}: A $\PRFSPD$ scheme satisfies $\secproof$ if for any $\qpt$ adversary $\adv$ in cloning game (see Game~\ref{game:cloning-prfspd}), there exists a negligible function $\negl$ such that
      \[\Pr[\clon{\adv}{\PRFSPD}=1]=\negl.\]
      \label{item:security of proof of independence}
\begin{game}
\caption{$\clon{\adv}{\PRFSPD}$ %
}
\begin{algorithmic}[1]
    \State Given input $1^\secpar$, Challenger samples $k\gets \{0,1\}^{w(\secpar)}$ uniformly at random.
    \State Initialize an empty set of variables, $S$.
    \State $\adv$ gets oracle access to $\gen(k,\cdot)$,  $\ver(k,\cdot, \cdot)$ as oracle.
    \For{$\gen$ query $x$ made by $\adv$}
    \If{$\exists$ variable $t_x\in S$}
    $t_x=t_x+1$.
    \Else{ Create a variable $t_x$ in $S$, initialized to $1$.}
    \EndIf
    \EndFor
    \State $\adv$ outputs $x,c_1,c_2,\ldots,c_{t_x+1}$ to the challenger.
    \State Challenger rejects if $c_i$'s are not distinct.
    \For{$i\in[m+1]$}
    $b_i\gets \ver(k,x,c_i)$
    \EndFor
    \State Return $\wedge_{i=1}^{m+1} b_i$.
\end{algorithmic}
\label{game:cloning-prfspd}
\end{game}

\end{enumerate}

\label{definition:PRFSPD}
\end{definition} %

\section{Definitions of qPKE}
\label{sect:defs}
In this section, we introduce the new notion of encryption with quantum public
keys (\cref{def:eqpk}).
The indistinguishability security notions are defined in~\cref{sect:def-eo}
and~\cref{sect:def-wo-eo}.

\begin{definition}[Encryption with quantum public keys]
\label{def:eqpk}
Encryption with quantum public keys (qPKE) consists of 4 algorithms with the
following syntax:
\begin{enumerate}
  \item \(\dk \gets \Gen(\secparam)\): a \(\qpt\) algorithm, which receives the security
    parameter and outputs a classical decryption key.
  \item \(\ket{\qpk}\gets \qpkgen(\dk)\): a \(\qpt\) algorithm, which receives a
    classical decryption key \(\dk\), and outputs a quantum public key
    \(\ket{\qpk}\).
    In this work, we require that the output is a pure state, and that \(t\)
    calls to \(\qpkgen(\dk)\) should yield the same state, that is,
    \(\ket{\qpk}^{\tensor t} \).
  \item \((\qpk',\qc) \gets \qenc(\qpk,m)\): a \(\qpt\) algorithm, which receives a
    quantum public key \(\qpk\) and a plaintext \(m\), and outputs a (possibly
    classical) ciphertext \(\qc\) and a recycled public key~\(\qpk'\).
  \item \(m \gets \qdec(\dk, \qc)\): a \(\qpt\) algorithm, which uses a decryption
    key \(\dk\) and a ciphertext \(\qc\), and outputs a classical plaintext
    \(m\).
\end{enumerate}
\end{definition}

We say that a qPKE scheme is \emph{correct} if for every message
\(m \in \{0,1\}^*\) and any security parameter \(\secpar \in \NN\), the following
holds:
\[
\Pr\left[
\qdec(\dk, \qc) = m
\ \middle\vert
\begin{array}{c}
\dk\gets \Gen(\secparam) \\
\ket{\qpk} \gets \qpkgen(\dk) \\
(\qpk', \qc) \gets \qenc(\ket{\qpk}, m)
\end{array}
\right] \geq 1 -\negl,
\]
where the probability is taken over the randomness of \(\Gen\), \(\qpkgen\),
\(\qenc\) and \(\qdec\).
We say that the scheme is reusable if completeness holds to polynomially many messages using a single quantum public key. More precisely, we say that a qPKE scheme is \emph{reusable} if for every security parameter \(\secpar \in \NN\), polynomial number of messages \(m_1,\ldots,m_{n(\lambda)} \in \{0,1\}^*\), the following
holds:
\[
\Pr\left[
\forall i \in [n(\secpar)],\, \qdec(\dk, \qc_i) = m_i
\ \middle\vert
\begin{array}{c}
\dk\gets \Gen(\secparam) \\
\ket{\qpk_1} \gets \qpkgen(\dk) \\
(\qpk_2, \qc_2) \gets \qenc(\ket{\qpk_1}, m_1)\\
 \vdots\\
(\qpk_{n+1}, \qc_n) \gets \qenc(\ket{\qpk_i}, m_{n(\secpar)})\\
\end{array}
\right] \geq 1 -\negl.
\]
\subsection{Security Definitions for qPKE with Classical Ciphertexts}
\label{sect:def-eo}
In this section, we present a quantum analogue of classical indistinguishability
security for qPKE with classical ciphertexts.
We note that there are few differences.
Firstly, since in general the public keys are quantum states and unclonable, in
the security games, we allow the adversary to receive polynomially many copies
of \(\ket{\qpk}\), by making several calls to the \(\qpkgen(\dk)\) oracle.
Secondly, in the classical setting, there is no need to provide access to an
encryption oracle since the adversary can use the public key to apply the
encryption themself.
In the quantum setting, this is not the case: as we will see, the quantum public
key might be measured, and the ciphertexts might depend on the measurement
outcome.
Furthermore, the quantum public key can be reused to encrypt multiple different
messages.
This motivates a stronger definition of indistinguishability with encryption
oracle, in which the adversary gets oracle access to the encryption, denoted as
IND-ATK-EO security, where ATK can be either chosen-plaintext attacks (CPA),
(adaptive or non-adaptive) chosen-ciphertext attacks (CCA1 and CCA2).

\begin{game}
  \caption{Indistinguishability security with an encryption oracle (IND-ATK-EO)
    for encryption with quantum public key and classical ciphertext schemes.}
  \label{game:ind-eo}
\begin{algorithmic}[1]
  \State The challenger generates \(\dk \gets \Gen(\secparam)\).
  \State The adversary gets \(\secparam\) as an input, and oracle access to
  \(\qpkgen(\dk)\).
  \State The challenger generates \(\ket{\qpk}\gets \qpkgen(\dk)\).
  Let
  \(\qpk_{1} \coloneqq \ket{\qpk}\).
  \State For \(i=1,\ldots,\ell\), the adversary creates a classical message \(m_i\) and
  send it to the challenger.
  \State The challenger computes \((\qc_i,\qpk_{i+1}) \gets \qenc(\qpk_i,m_i)\) and send \(\qc_i\) to the adversary.
  \State During step (2) to step (5), the adversary also gets classical oracle
  access to an oracle \(\ooracle_{1}\).
  \State The adversary sends two messages \(m'_0,m'_1\) of the same length to the challenger.
  \State The challenger samples \(b\in_R \{0,1\}\), computes
  \((\qc^{*}, \qpk_{l+2}) \gets \qenc(\qpk_{\ell+1},m'_b)\) and sends \(\qc^{*}\) to the adversary.
  \State For \(i=\ell+2,\ldots,\ell'\), the adversary creates a classical message \(m_i\) and send it to the challenger.
  \State The challenger computes \((\qc_i,\qpk_{i+1}) \gets \qenc(\qpk_i,m_i)\) and send \(\qc_i\) to the adversary.
  \State During step (9) to step (10), the adversary also gets classical oracle
  access to an oracle \(\ooracle_{2}\).
  Note that after step (7), the adversary no longer gets access to oracle \(\ooracle_{1}\).
  \State The adversary outputs a bit \(b'\).
\end{algorithmic}
We say that the adversary wins the game (or alternatively, that the outcome of
the game is 1) iff \(b=b'\).
\end{game}
We define the oracles \(\ooracle_{1}, \ooracle_{2}\) depending on the level of security as follows.

\begin{pchstack}[center]
\procedure{ATK}{%
    \text{CPA}  \\
    \text{CCA1} \\
    \text{CCA2} }

\procedure{Oracle $\ooracle_1$}{%
\t  \varnothing  \\
\t  \qdec(\dk, \cdot)  \\
\t  \qdec(\dk, \cdot)
}

\procedure{Oracle $\ooracle_2$}{%
\t  \varnothing  \\
\t  \varnothing  \\
\t  \qdec^*(\dk, \cdot)
}
\end{pchstack}
Here \(\qdec^{*}(\dk, \cdot)\) is defined as \(\qdec(\dk, \cdot)\), except that it
return \(\bot\) on input the challenge ciphertext \(\qc^{*}\).

\begin{definition}
  \label{def:ind-eo}
  A qPKE scheme is IND-ATK-EO secure if for every \(\qpt\) adversary, there
  exists a negligible function \(\epsilon\) such that the probability of winning the
  IND-ATK-EO security game (\cref{game:ind-eo}) is at most
  \(\frac{1}{2}+\epsilon(\secpar)\).
\end{definition}

\begin{remark}
  The definition presented in~\cref{def:ind-eo} is stated for the single
  challenge query setting.
  Using the standard hybrid argument, it is straightforward to show that
  single-challenge definitions also imply many-challenge definitions where the
  adversary can make many challenge queries.
\end{remark}

\begin{remark}
Note that the IND-CCA2-EO definition is only well-defined for schemes with classical ciphertexts. The other two notions are well-defined even for quantum ciphertexts, though we do not use those.
\end{remark}

\subsection{Security Definitions for qPKE with Quantum Ciphertexts}
\label{sect:def-wo-eo}
We now give a definition for qPKE with quantum ciphertexts.
In the case of adaptive chosen ciphertext security, the definition is
non-trivial due to the no-cloning and the destructiveness of quantum
measurements.
We note there are indeed several works considering the notions of chosen-ciphertext security in the quantum setting: \cite{EC:AlaGagMaj18} defines chosen-ciphertext security for quantum symmetric-key encryption (when the message is a quantum state), and \cite{C:BonZha13,CEV22} defines chosen-ciphertext security for classical encryption under superposition attacks.
However, extending the technique from \cite{EC:AlaGagMaj18} to the public-key setting is non-trivial, and we leave this open problem for future work.
In this section, we only consider security notions under chosen-plaintext
attacks and non-adaptive chosen-ciphertext attacks.

Even though one can similarly define security notions with encryption oracle for schemes with quantum ciphertexts as in~\cref{sect:def-eo}, we note that in all constructions of qPKE with quantum ciphertexts present
in this work are not reusable, and thus we do not present the definition in
which the adversary has oracle access to the encryption oracle for the sake of
simplicity.
We denote these notions as IND-ATK, where ATK is either chosen-plaintext attacks
(CPA) or non-adaptive chosen-ciphertext attacks (CCA1).

\begin{game}
  \caption{IND-ATK security game for encryption with quantum public key and
    quantum ciphertexts schemes.}
  \label{game:ind-atk}
  \begin{algorithmic}[1]
    \State The challenger generates \(\dk \gets \Gen(\secparam)\).
    \State The adversary \(\adv_{1}\) gets \(\secparam\) as an input, and oracle
    access to \(\qpkgen(\dk)\), $\qenc(\qpk,\cdot)$ and \(\ooracle_{1}\), and sends \(m_0,m_1\) of the
    same length to the challenger.
    \(\adv_{1}\) also output a state \(\ket{\mathsf{st}}\) and sends it to \(\adv_{2}\).
    \State The challenger samples \(b\in_R \{0,1\}\), generates
    \(\ket{\qpk}\gets \qpkgen(\dk)\) and sends \(c^{*} \gets \qenc(\ket{\qpk},m_b)\) to the
    adversary \(\adv_{2}\).
    \State \(\adv_{2}\) gets oracle access to \(\qpkgen(\dk)\),  $\qenc(\qpk,\cdot)$.
    \State The adversary \(\adv_{2}\) outputs a bit \(b'\).
  \end{algorithmic}
  We say that the adversary wins the game (or alternatively, that the outcome of
  the game is 1) iff \(b=b'\).
\end{game}
The oracles \(\ooracle_{1}\) is defined depending on the level of security as
follows.
\begin{pchstack}[center]
\procedure{ATK}{%
    \text{CPA}  \\
    \text{CCA1}
    }

\procedure{Oracle $\ooracle_1$}{%
\t  \varnothing  \\
\t  \qdec(\dk, \cdot)
}

\end{pchstack}

\begin{definition}
  \label{def:ind-atk}
  A qPKE scheme with quantum ciphertexts is IND-ATK secure if for every \(\qpt\)
  adversary \(\adv \coloneqq (\adv_{1}, \adv_{2})\), there exists a negligible
  function \(\epsilon\) such that the probability of winning the IND-ATK security game
  (\cref{game:ind-atk}) is at most \(\frac{1}{2}+\epsilon(\secpar)\).
\end{definition}

\section{Constructions of CCA-Secure qPKE}
\label{sect:constructions}
In this section, we present our qPKE constructions from OWF and PRFS and prove
that their CCA security.
The former (given in~\cref{sec:cca_qpke}) has classical ciphertexts, and allows
to encrypt arbitrary long messages.
The latter (given in~\cref{sect:cca_from_prfs}) has quantum ciphertexts, and
only allows to encrypt a single-bit message.
However, we note that the latter is based on a weaker assumption than the
former.
Finally, in~\cref{sec:nm-cpa}, we give a remark on the black-box construction of
non-malleable qPKE from CPA-secure qPKE using the same classical approach.

\subsection{CCA-Secure Many-Bit Encryption from OWF}
\label{sec:cca_qpke}

We start by presenting a simple qPKE construction from OWF which prove that it provides our strongest notion of security, i.e. IND-CCA-EO security. The scheme is formally presented in \cref{fig:qkem}. The ciphertexts produced by the scheme are classical, and the public-keys are reusable. The cryptographic components of our construction are a quantum secure PRF family $\{f_{k}\}$ and a post-quantum IND-CCA secure symmetric-key encryption scheme $(\SE.\enc,\SE.\dec)$ which can both be built from a quantum-secure OWF~\cite{FOCS:Zhandry12,EC:BonZha13}.

\begin{construction}[IND-CCA-EO secure qPKE from OWF]
  \label{fig:qkem}
  \begin{itemize}[label=\(\bullet\),itemsep=1ex]
    \item \textbf{Assumptions:} A family of quantum-secure pseudorandom
      functions $\{f_k\}_k$, and post-quantum IND-CCA SKE $(\SE.\enc , \SE.\dec)$.
    \item $\underline{\Gen(\secparam)}$
      \begin{compactenum}
        \item $\dk \xleftarrow{\$} \lbrace 0, 1 \rbrace^\lambda$
        \item $\ket{\qpk} \gets \sum_{x \in \lbrace 0, 1 \rbrace^\lambda} |x, f_{\dk}(x) \rangle$
      \end{compactenum}
    \item $\underline{\qenc(\ket{\qpk}, m)}$
      \begin{compactenum}
        \item Measure $\ket{\qpk}$ to obtain classical strings $x, y$.
        \item Let $c_0 \gets x$ and $c_1 \gets \SE.\enc(y,m)$.
        \item Output $(c_0, c_1)$
      \end{compactenum}
    \item $\underline{\qdec(\dk,(c_0, c_1))}$
      \begin{compactenum}
        \item Compute $y \gets f_{\dk}(c_0)$.
        \item Compute $m \gets \SE.\dec(y,c_1)$ and return $m$.
      \end{compactenum}
  \end{itemize}
\end{construction}

It can be trivially shown that the scheme achieves perfect correctness if the underlying SKE provides the perfect correctness property. 

\begin{theorem}
  \label{thm:kem_cca}
  Let $\{f_{k}\}_{k}$ be a quantum secure PRF and $(\SE.\enc,\SE.\dec)$ be a post-quantum IND-CCA secure SKE. Then, the quantum qPKE given in~\cref{fig:qkem} is IND-CCA-EO secure.
\end{theorem}

\begin{proof}
  We proceed with a sequence of hybrid games detailed in

  \begin{itemize}
    \item \textbf{Hybrid $H_0$:} This is the IND-CCA game with $\Pi$ with the challenge ciphertext fixed to $(x^*,c^*) = \qenc(\ket{\pk},m_0')$.
    \item \textbf{Hybrid $H_1$:} This is identical to $H_0$ except instead of measuring $\ket{\qpk}$ when the adversary queries the encryption oracle, the challenger measures a copy of $\ket{\qpk}$ in advance to obtain $(x^*,y^*=f_{\dk}(x^*))$ and answers queries to the encryption oracle using $(x^*,y^*)$ instead. The decryption oracle still returns $\bot$ when queried $(x^*,c^*)$. This change is only syntactical so the two hybrids are the same from the adversary's view.
    \end{itemize} 
    The hybrids $H_2$ to $H_5$ have 2 main goals: (i) to decorrelate the encryption/decryption oracles $\qdec^*, \qenc$ from the public-keys handed to the adversary and (ii) to remove the oracles' dependency on $\dk$. 
    \begin{itemize}
    \item \textbf{Hybrid $H_2$:} This is identical to $H_1$, except $(x^*,y^*)$ is removed from the copies of $\ket{\qpk}$ handed to the adversary. More precisely, the adversary is handed $\ket{\qpk'}$ of the following form:
    \begin{equation}
        \ket{\qpk'} = \frac{1}{\sqrt{2^{|x^*|}-1}}\sum_{x:x\neq x'} \ket{x}\ket{f_{\dk}(x)}
    \end{equation}
    The decryption oracle still returns $\bot$ when queried on the challenge ciphertext. Note that $\ket{\qpk}$ and $\ket{\qpk'}$ have $\negl$ trace distance so the advantage of distinguishing $H_1$ and $H_2$ is $\negl$. 
    \item \textbf{Hybrid $H_3$:} This (inefficient) hybrid is identical to $H_2$ other than $f_{\dk}$ being replaced with a truly random function $f$, i.e. the public-keys are change to:
    \begin{equation}
        \ket{\qpk'} = \frac{1}{\sqrt{2^{|x^*|}-1}}\sum_{x:x\neq x'} \ket{x}\ket{f(x)}
    \end{equation} 
    The encryption and decryption oracle can be simulated by oracle access to $f$. The decryption oracle returns $\bot$ when queried $(x^*,c^*)$. The indistinugishability of $H_3$ and $H_2$ follows directly from pseudorandomness property of $\{f_k\}_k$.
    \item\textbf{Hybrid $H_4$:} This hybrid is identical to $H_3$ other than $y^*$ being sampled uniformly at random. Upon quering $(c_0,c_1)$ to the decryption oracle if $c_0 \neq x^*$, the oracle computes $y = f(c_0)$ and returns $m = \SE.\dec(y,c_1)$. In case $c_0 = x^*$ and $c_1\neq c^*$, the decryption oracle  returns $m = \SE.\dec(y^*,c_1)$. On $(x^*,c^*)$ the oracle returns $\bot$. The encryption oracle returns $(x^*,\SE.\enc(y^*,m))$ when queried on $m$. As $x^*$ does not appear in any of the public-keys this change is only syntactical. 
    \item\textbf{Hybrid $H_5$:} This hybrid reverts the changes of $H_3$, i.e. $\dk'$ is sampled uniformly at random and the public-keys are changed as follows:
    \begin{equation}
        \ket{\qpk'} = \frac{1}{\sqrt{2^{|x^*|}-1}}\sum_{x:x\neq x'} \ket{x}\ket{f_{\dk'}(x)}
    \end{equation} 
     With this change, on query $(c_0,c_1)$ if $c_0 \neq x^*$, the decryption oracle computes $y = f_{\dk'}(c_0)$ and returns $m = \SE.\dec(y,c_1)$. In case $c_0 = x^*$, the decryption oracle simply returns $m = \SE.\dec(y^*,c_1)$ when $c_1\neq c^*$ and $\bot$ otherwise. The encryption oracle is unchanged from $H_4$. The indistinguishability of $H_4$ and $H_5$ follows from the pseudorandomness of $f$ and the fact that $\ket{\qpk'}$ and $(x^*,y^*)$ are decorrelated. The hybrid is efficient again. 
     \end{itemize}
     The next step is to remove the dependency of the encryption and decryption oracles on $y^*$. This is done by querying the encryption and decryption oracles of the SKE. 
     \begin{itemize}
     \item\textbf{Hybrid $H_6$:} Let $\SE.\oenc$ and $\SE.\odec^*$ be two oracles implementing the encryption and decryption procedures of $\SE$ with the key $y^*$. $\SE.\odec^*$ returns $\bot$ when queried $y^*$. In this hybrid, we syntactically change the encryption and decryption oracle using these two oracles. To implement the encryption oracle, on query $m$ we simply query $\SE.\oenc$ on message $m$ and return $(x^*, \SE.\oenc(m))$. To simulate the decryption oracle, on query $(c_0,c_1)$ we act the same as in $H_5$ when $c_0\neq x^*$, but on queries of form $(x^*,c)$ we query $\SE.\odec^*$ on $c$ and return $\SE.\odec^*(c)$. Due to the definition of $\oenc$ and $\odec^*$ these changes are also just syntactical. Note that although $\SE.\odec^*$ always returns $\bot$ on $y^*$, it is only queried when $c_0 = x^*$, i.e. to cause this event the decryption oracle should be queried on the challenge ciphertext $(x^*,c^*)$. 

     \item\textbf{Hybrid $H_7$:} We provide the adversary with $x^*, \SE.\oenc,\SE.\odec^*$, instead of access to the encryption and decryption oracle. Note that the adversary can implement the encryption and decryption oracles themselves by having access to $x^*,\SE.\oenc,\SE.\odec^*$ and sampling a uniform $\dk'$ themselves and vice versa ($\SE.\odec^*$ can be queried on $c$ by querying the decryption oracle $(x^*,c)$ and $\SE.\oenc$ can be queried on $m$ by querying the encryption oracle on $m$). This demonstrates that the hybrids are only syntactically different and hence are indistinguishable. 

     \item\textbf{Hybrid $H_8$:} This hybrid is identical to $H_7$ with the only difference that the challenge ciphertext is swapped with $(x^*, \SE.\oenc(0))$. Now notice that any adversary that can distinguish $H_8$ from $H_7$ can effectively break the IND-CCA security of $\SE$. Hence, the indistinguishability of the two hybrids follows directly from the IND-CCA security of $\SE$. 

     Following the same exact hybrids for challenge ciphertext $\qenc(\ket{\qpk}, m_1')$ we can deduce that the scheme is IND-CCA-EO secure. 
    
    \end{itemize}
  \end{proof}

\subsection{CCA1-Secure Many-Bit Encryption from PRFS}
\label{sect:cca_from_prfs}
We continue by presenting a CCA1-secure bit-encryption from PRFS. Extending this scheme to polynomially many bits is discussed at the end of this section, see \cref{rem:length-unrestricted}. The description of the scheme is given below
in~\cref{fig:cca_prfs}.

\begin{construction}[IND-CCA1 secure qPKE from PRFS]
  \label{fig:cca_prfs}
  \begin{itemize}[label=\(\bullet\),itemsep=1ex]
    \item\textbf{Assumptions:} A PRFS family \(\{\ket{\psi_{\dk,x}}\}_{\dk,x}\)
      with super-logarithmic input-size.
      Let \(n \coloneqq n(\secpar)\).

    \item \(\underline{\Gen(\secparam)}\)
      \begin{compactenum}
        \item Output \(\dk \gets_R \{0,1\}^\secpar\).
      \end{compactenum}

    \item \(\underline{\qpkgen(\dk)}\)
      \begin{compactenum}
        \item Output
        \(\ket{\qpk} \gets \sum_{x} \ket{x}_{R} \ket{\psi_{\dk, x}}^{\otimes n}_{S}\), where
        \(x \in \bin^{\omega(\log \secpar)}\).
      \end{compactenum}

    \item \(\underline{\qenc(\ket{\qpk}, m)}\) for \(m \in \bin\)
      \begin{compactenum}
        \item Measure the \(R\) registers of \(\ket{\qpk}\) to obtain a
        classical string $x$.
        Let \(\ket{\phi} \coloneqq \ket{\psi_{\dk, x}}^{\otimes n}\) denote the residual
        state.
        \item If \(m = 0\), output the ciphertext as \((x, \ket{\phi})\).
        \item Else, sample a uniformly random key $\dk_1$, and output the ciphertext as $(x, \ket{\psi_{\dk_1, x}}^{\otimes n})$.\label{line:samp_haar}
      \end{compactenum}

    \item \(\underline{\qdec(\dk, (x, \Psi))}\)
      \begin{compactenum}
        \item Compute \(\ket{\psi_{\dk, x}}^{\otimes n}\) and perform $n$ SWAP tests for each subsystem of $\Psi$ of the same size as $\ket{\psi_{\dk,x}}$ with $\ket{\psi_{\dk,x}}$. 
        \item If the outcome of the SWAP tests is \(0\) all the time, output
        \(0\), otherwise output \(1\).
      \end{compactenum}
  \end{itemize}
\end{construction}

The correctness of the scheme follows from the fact that the states $\ket{\psi_{\dk_1,x}}$ are relatively well spread out for a random choice of $\dk$. This is due to the pseudorandomness of the state generator. Note that if in step~\ref{line:samp_haar} instead of picking $\dk_1$ randomly and computing $\ket{\psi_{\dk_1,x}}$, the encryption algorithm sampled $\ket{\vartheta}^{\ot n}$, from the Haar measure, the expected probability of $n$ SWAP tests between  $\ket{\psi_{x,\dk}}$ and $\ket{\vartheta}$ all returning $0$ would be $2^{-n}$. Hence, if the probability is more than negligibly apart for $n$ SWAP tests between $\ket{\psi_{x,\dk_1}}$ and $\ket{\psi_{x,\dk}}$ for a random choice of $\dk_1$, with a Chernoff bound argument one can show that this would lead to a distinguisher for the PRFS. Hence, for $n$ polynomial in $\lambda$ the scheme has negligible correctness error.
\begin{theorem}
\label{thm:prfs_cca1}
  The construction in \cref{fig:cca_prfs} is IND-CCA1
  secure (see \cref{def:ind-atk}), assuming \(\{\ket{\psi_{\dk,x}}\}_{\dk,x}\)
  is a PRFS with super-logarithmic input-size.
\end{theorem}
\begin{proof}
  We prove the theorem via a series of hybrids.
  \begin{itemize}
    \item \textbf{Hybrid \(H_0\).}
      The original security game as defined in~\cref{def:ind-atk}.

    \item \textbf{Hybrid \(H_1\).}
      This is identical to hybrid \(H_0\), except that the challenger, instead
      of measuring \(\ket{\qpk}\) when the adversary queries the encryption
      oracle for the first time, the challenger measures (the \(R\) registers
      of) this state before providing the copies of \(\ket{\qpk}\) to the
      adversary.
      Note that by measuring \(\ket{\qpk}\) in the computational basis, the
      challenger would obtain a classical uniformly random string
      \(x^*\), let the residual state be
      \(\ket{\phi^*} \coloneqq \ket{\psi_{\dk, x^*}}^{\otimes n}\).

      Note that the two operations corresponding to the challenger's measurement
      of \(\ket{\qpk}\) and the creation of the copies of \(\ket{\qpk}\) given
      to the adversary commute.
      Thus, the distribution of the two hybrids are identical and no adversary
      can distinguish \(H_{0}\) from \(H_{1}\) with non-zero advantage.

    \item \textbf{Hybrid \(H_2\).}
      This is identical to hybrid \(H_1\), except that the challenger samples
      \(x^*\) as in the previous hybrid, and
      instead of providing \(\ket{\qpk}\) to the adversary, it provides
\[\ket{\qpk'} \coloneqq \frac{1}{\sqrt{2^{\size{x^*}} -1}}\sum_{x: x \neq x^*}  \ket{x}\ket{\psi_{\dk, x}}^{\otimes n}.\]
      Moreover, in the challenge query, the challenger uses
      \((x^*, \ket{\phi^*})\) for the encryption of the chosen
      message \(m\), without measuring a fresh copy of \(\ket{\qpk}\) (that is,
      it skips the first step of the encryption algorithm).
      We note that this state \(\ket{\qpk'}\) can be efficiently prepared.

      The distinguishing probability of the two hybrids \(H_{1}\) and \(H_{2}\)
      implies that we can distinguish the following quantum states
      \(\ket{\qpk}^{\otimes p} \otimes \ket{x^*}\) and
      \(\ket{\qpk'}^{\otimes p} \otimes \ket{x^*}\)
      with the same probability, but these two quantum states have \(\negl\)
      trace-distance for any polynomial \(p\).
      Therefore, any adversary can only distinguish \(H_{1}\) and \(H_{2}\) with
      success probability at most \(\negl\).

    \item \textbf{Hybrid \(H_{3}\).}
      This (inefficient) hybrid is identical to \(H_{2}\), except that the
      challenger uses a Haar oracle \(\mathcal{O}_{\mathsf{Haar}}\) to generate
      \(\ket{\qpk'}\) in place of \(\ket{\psi_{\dk, \cdot}}\).
      In particular, the quantum public key in the hybrid \(H_{3}\) is computed
      as:
      \begin{equation*}
        \ket{\qpk'} \gets \sum_{x: x \neq x^*} \ket{x} \otimes \ket{\vartheta_{x}}^{\otimes n},
      \end{equation*}
      where each \(\ket{\vartheta_{x}}\) is an output of \(\mathcal{O}_{\mathsf{Haar}}\) on input \(x\).
      The decryption oracle is the same as the decryption algorithm with the difference that $\ooracle_{\mathsf{PRFS}}$ (the algorithm generating the PRFS) is swapped with $\ooracle_{\mathsf{Haar}}$. The crucial point here is that the decryption oracle only uses the PRFS in
      a black-box way (in particular, it only uses \(\mathcal{O}_{\mathsf{PRFS}}\) and
      does not use \(\mathcal{O}_{\mathsf{PRFS}}^{\dagger}\)).
      
      Note that the decryption oracle can return $\bot$ on query $(x^*,\cdot)$. This can not be used to distinguish the two hybrids as the adversary has a negligible chance of querying $x^*$ as $x^*$ is picked uniformly at random. The adversary is only provided with the value of $x^*$ when given the challenge ciphertext, at which point they do not have access to the decryption oracle anymore. 
      
      We note that the adversary does not have direct access to this $\ooracle_\mathsf{Haar}$, but only via the decryption oracle. By pseudorandomness property of \(\ket{\psi_{\dk, \cdot}}\), we have that
      \(H_{2}\) and \(H_{3}\) are computationally indistinguishable.

    \item \textbf{Hybrid \(H_{4}\).}
      In this hybrid, we revert the changes in \(H_{3}\), except that the challenger samples a
      uniformly random key \(\dk'\) to compute all states in $\ket{\qpk'}$, except for the one used to encrypt the challenge query.
      In particular, the public key $\ket{\qpk'}$ is now generated using the PRFS generator with the key $\dk'$, and the secret key $\dk$ and its public counterpart \((x^{*}, \ket{\psi_{\dk, x^{*}}}^{\otimes n})\) are used for the challenge encryption.
      We note that the hybrid is now efficient again.
      Similar to the previous argument, $H_3$ and $H_4$ are also computationally indistinguishable due to pseudorandomness property of $\ket{\psi_{\dk', \cdot}}$.

    \item \textbf{Hybrid $H_5$.}
      This hybrid is identical to $H_4$, except that in the challenge query, instead of encrypting $0$ as $(x^{*}, \ket{\psi_{\dk, x^{*}}}^{\otimes n})$, the challenger encrypts $0$ as $(x^{*}, \ket{\vartheta_{x^{*}}}^{\otimes n})$, where each \(\ket{\vartheta_{x}}\) is an output of \(\mathcal{O}_{\mathsf{Haar}}\) on input \(x\).
      
      Notice that in this hybrid, the secret key $\dk$ and its public counterpart $(x^*, \ket{\psi_{\dk, x^*}})^{\otimes n})$ are not correlated with any of other variables in the hybrid.
      Furthermore, after receiving the challenge ciphertext, the adversary no longer gets access to the decryption oracle. 
      By the pseudorandomness property of $\ket{\psi_{\dk, x^*}}$, we have that $H_4$ and $H_5$ are computationally indistinguishable.

      Furthermore, in this final hybrid, the adversary needs to distinguish the output of PRFS with a uniformly random key $\dk_1$ (for encryption of 1) and the output of a Haar random oracle (for encryption of 0).
      By the same argument as above, the winning advantage of the adversary is also negligible.

  \end{itemize}
  Overall, since all hybrids are negligibly close and the winning advantage of the adversary in the last hybrid in negligible, we conclude the proof.
\end{proof}

\begin{remark}
We sketch here how to achieve many-bit encryption (i.e., non-restricted length encryption) from our scheme present above.
We do this through several steps.
\begin{itemize}
    \item The scheme stated in \cref{fig:cca_prfs} can easily be extended to a length-restricted scheme, by applying bit-by-bit encryption.
    \item Given a qPKE length-restricted CCA1 encryption, and a (non-restricted length) symmetric key encryption, we can define a hybrid encryption scheme, where the qPKE scheme is used first to encrypt a random (fixed length) secret key, which is later used to encrypt an arbitrarily long message. The entire scheme is CPA- (respectively, CCA1-) secure if the symmetric key encryption has CPA- (respectively, CCA1-) security.
    \item Finally, we note that the following many-bit symmetric key encryption scheme can be proven to be CCA1 secure, using the same proof strategy as in~\cref{thm:prfs_cca1}, based on the existence of PRFS alone.
    Given a secret key $\dk$, to encrypt a message $m \in \bin^{\ell}$, we sample $\ell$ distinct uniformly random strings $x_i$, and compute $\ket{\psi_{\dk, x_i}}^{\otimes n}$. Then each bit $m_i$ will be encrypted using as $(x_i, \ket{\psi_{\dk, x_i}}^{\otimes n})$ if $m_i = 0$, or $(x_i, \ket{\psi_{\dk', x_i}}^{\otimes n})$ if $m_i = 1$ for a fresh key $\dk'$.
\end{itemize}
\label{rem:length-unrestricted}
\end{remark} %
\subsection{Generic Construction of Non-Malleable qPKE}
\label{sec:nm-cpa}

We remark that known implications from the literature can be used to show that
IND-CPA secure qPKE \emph{with classical ciphertexts} implies non-malleable
qPKE: The work of~\cite{JC:CDMW18} shows a black-box compiler from IND-CPA
encryption to non-malleable encryption, which also applies to the settings of
quantum public-keys.
The only subtlety is that the compiler assumes the existence of a one-time
signature scheme to sign the ciphertext.
In~\cite{C:MorYam22,cryptoeprint:2022/1336} it is shown that one-time signatures
(with quantum verification keys) exist assuming one-way state generators, which
in turn are implied by qPKE.
Combining the implications of these two works, we obtain a generic construction
of non-malleable qPKE from any IND-CPA secure one.

\section{IND-CPA-EO secure qPKE from PRFSPD}
\label{sect:qpke-from-prfspd}
In this section, we propose a construction for qPKE from pseudo-random
function-like states with proof of destruction. The construction is reusable, has classical ciphers, and is CPA-EO secure.

We first import the following result that builds \emph{symmetric}-key
encryption from PRFSPD.

\begin{proposition}[\cite{cryptoeprint:2023/543}]
  If quantum-secure PRFSPD exists, then there exists a quantum CPA symmetric
  encryption with classical ciphertexts.
\end{proposition}

We give the formal construction for many-bit reusable encryption scheme from PRFSPD in~\cref{fig:public_key_encryption_from_prfspd}.

\begin{construction}[IND-CPA-EO secure qPKE from PRFSPD]
  \label{fig:public_key_encryption_from_prfspd}
  \begin{itemize}[label=\(\bullet\),itemsep=1ex]
    \item \textbf{Assumptions:} A PRFSPD family \(\{\ket{\psi_{\dk,x}}\}_{\dk,x}\)
      and a quantum symmetric encryption scheme with classical ciphers
      \(\{\enc,\dec\}\).
    \item \(\underline{\Gen(\secparam)}\)
      \begin{compactenum}
        \item Let \(\dk_{0, i} \gets_R \{0,1\}^\secpar\) for all
        \(i \in [1, \secpar]\).
        \item Output \(\dk \gets \{\dk_{0, i}\}_{i \in [1, \secpar]}\).
      \end{compactenum}

    \item \(\underline{\qpkgen(\dk)}\)
      \begin{compactenum}
        \item Output
        \( \ket{\qpk}=\bigotimes_{i \in [\lambda]} \frac{1}{\sqrt{2^\secpar}} \sum_{x_{i} \in \{0,1\}^\secpar}\ket{x_{i}}\ket{\psi_{\dk_{0, i},x_{i}}}\).
      \end{compactenum}

    \item \(\underline{\qenc(\ket{\qpk},m)}\) for \(m \in\{0,1\}^*\)
      \begin{compactenum}
        \item Let
        \(\ket{\qpk_i} \coloneqq \frac{1}{\sqrt{2^\secpar}} \sum_{x_{i} \in \{0,1\}^\secpar}\ket{x_{i}}\ket{\psi_{\dk_{0, i},x_{i}}}\),
        and write \(\ket{\qpk}\) as
        \(\ket{\qpk}=\bigotimes_{i \in [\lambda]} \ket{\qpk_i}\).
        \item Measure the left registers of \(\ket{\qpk_i}\) to obtain classical
        strings \(x_{i}\).
        Denote the post-measurement states as \(\ket{\psi'_i}\).
        \item Set \(y_{i}\gets \destruct(\ket{\psi'_{i}})\).
        \item For \(i \in [1, \secpar]\), pick \(\dk_{1, i} \gets \{0,1\}^\lambda\) and
        compute $\ket{\psi_{\dk_{1, i}, x_{i}}}$.
        \item Set \(y'_{i}\gets \destruct(\ket{\psi_{\dk_{1, i}, x_{i}}})\) for all
        $i \in [\secpar]$.
        \item Pick a uniformly random key $k \gets \bin^\secpar$.
        \item Set \(\tilde{y}_{i} = \begin{cases}y'_{i}&, \text{if
          } k_i = 0 \\ y_{i}&, \text{if } k_i = 1 \end{cases}\).
        \item Output
        \(\left(\enc(k, m), \left( (x_{i},\tilde{y}_{i})\right)_i\right)\) as
        ciphertext and \(\left(k, \left( (x_{i},\tilde{y}_{i})\right)_i\right)\)
        as the recycled public-key.
      \end{compactenum}

    \item \(\underline{\qdec(\dk, c)}\)
      \begin{compactenum}
        \item Interpret \(c\) as
        \(\left(c', \left( (x_{i},\tilde{y}_{i})\right)_i\right)\).
        \item Let \(k'_i = \ver(\dk_{0, i}, x_{i},\tilde{y}_{i})\) and let
        \(k' = k'_{0} \ldots k'_{\secpar}\).
        \item Output \(\dec(k', c')\).
      \end{compactenum}
  \end{itemize}
\end{construction}

The correctness of our scheme relies on the existence of PRFSPD with pseudorandomness and unclonability of proofs properties. 
The proof of correctness can be shown similarly to that of~\cref{fig:cca_prfs}.
Next, we show that this construction achieves IND-CPA-EO security in~\cref{lem:security-prfspd-scheme}.

\begin{theorem}\label{lem:security-prfspd-scheme}
  If quantum-secure PRFSPD with super-logarithmic input size exists, then there
  exists a public-key encryption with classical ciphertexts which is IND-CPA-EO
  secure.
\end{theorem}
\begin{proof}
Our construction is given in \cref{fig:public_key_encryption_from_prfspd}. It uses a PRFSPD family \(\{\ket{\psi_{\dk,x}}\}_{\dk,x}\) and a quantum symmetric encryption scheme with classical ciphers \(\{\enc,\dec\}\).
We prove the security of our scheme through a series of hybrids.

\begin{itemize}
  \item \textbf{Hybrid \(H_0\).}
    The original security game as defined in~\cref{def:ind-eo}.

  \item \textbf{Hybrid \(H_1\).}
    This is identical to hybrid \(H_0\), except that the challenger, instead of
    measuring \(\ket{\qpk_{i}}\) (for all \(i \in [\secpar]\)) when the adversary
    queries the encryption oracle for the first time, the challenger measures
    the left register of each \(\ket{\qpk_{i}}\) before providing the copies of
    \(\ket{\qpk}\) to the adversary.
    Note that by measuring \(\ket{\qpk_{i}}\) in the computational basis, the
    challenger would obtain a classical uniformly random string \(x^*_{i}\).

    Note that the two operations corresponding to the challenger's measurement
    of \(\ket{\qpk}\) and the creation of the copies of \(\ket{\qpk}\) given to
    the adversary commute.
    Thus, the distribution of the two hybrids are identical and no adversary can
    distinguish \(H_{0}\) from \(H_{1}\) with non-zero advantage.

  \item \textbf{Hybrid \(H_2\).}
    This is identical to hybrid \(H_1\), except that the challenger samples
    \(x^*_{i}\) as in the previous hybrid, and instead of providing
    \(\ket{\qpk}\) to the adversary, it provides
    \[\ket{\qpk'} \coloneqq \bigotimes_{i \in [\secpar]} \frac{1}{\sqrt{2^{\size{x^*_{i}}} -1}}\sum_{x_{i}: x_{i} \neq x^*_{i}} \ket{x_{i}}\ket{\psi_{\dk_{0, i}, x_{i}}}.\]
    Moreover, in the challenge query, the challenger uses
    \((x^*_{i}, \ket{\psi_{\dk_{0, i}, x^*_{i}}})\) for all \(i \in [\secpar]\) for
    the encryption of the chosen message \(m\), without measuring a fresh copy
    of \(\ket{\qpk}\) (that is, it skips the first step of the encryption
    algorithm).
    We note that this state \(\ket{\qpk'}\) can be efficiently prepared.

    The distinguishing probability of the two hybrids \(H_{1}\) and \(H_{2}\)
    implies that we can distinguish the following quantum states
    \(\ket{\qpk}^{\otimes p} \otimes \bigotimes_{i \in [\secpar]} \ket{x^*_{i}}\) and
    \(\ket{\qpk'}^{\otimes p} \otimes \bigotimes_{i \in [\secpar]} \ket{x^*_{i}}\) with
    the same probability, but these two quantum states have \(\negl\)
    trace-distance for any polynomial \(p\).
    Therefore, any adversary can only distinguish \(H_{1}\) and \(H_{2}\) with
    success probability at most \(\negl\).

  \item \textbf{Hybrid \(H_{2, i}\) for \(i \in [0, \secpar]\).}
    We define a series of (inefficient) hybrids \(H_{2, i}\), in which
    \(H_{2, 0} \coloneqq H_{2}\), and we denote
    \(H_{2, \secpar} \coloneqq H_{3}\).
    Each \(H_{2, i+1}\) is identical as \(H_{2, i}\), except that the challenger
    uses a Haar oracle \(\mathcal{O}_{\mathsf{Haar}_{i}}\) in place of
    \(\ket{\psi_{\dk_{0, i}, \cdot}}\).
    In particular, the quantum public key in the hybrid \(H_{2, i}\) is computed
    as:
      \begin{equation*}
        \ket{\qpk'} \gets \bigotimes_{j = 1}^{i} \sum_{x_{j}: x_{j} \neq x^*_{j}} \ket{x_{j}} \otimes \ket{\vartheta_{x_{j}}} \otimes \bigotimes_{j = i+1}^{\secpar} \sum_{x_{j}: x_{j} \neq x^*_{j}} \ket{x_{j}}\ket{\psi_{\dk_{0, j}, x_{j}}},
      \end{equation*}
    where each \(\ket{\vartheta_{x_{j}}}\) is an output of
    \(\mathcal{O}_{\mathsf{Haar}_{j}}\) on input \(x_{j}\).
    For the challenge encryption query, the challenger uses
    \((x^*_{j}, \ket{\vartheta_{x^*_{j}}})\) for all \(j \in [1, i]\), and
    \((x^*_{j}, \ket{\psi_{\dk_{0, j}, x^*_{j}}})\) for all \(j \in [i+1, \secpar]\).

    By pseudorandomness property of \(\ket{\psi_{\dk_{0, i}, \cdot}}\), we have that
    \(H_{2, i}\) and \(H_{2, i+1}\) are computationally indistinguishable.

  \item \textbf{Hybrid \(H_{3, i}\) for \(i \in [0, \secpar]\).}
    We define a series of (inefficient) hybrids \(H_{3, i}\), in which
    \(H_{3, 0} \coloneqq H_{3}\), and we denote
    \(H_{3, \secpar} \coloneqq H_{4}\).
    In each \(H_{3, i+1}\), we revert the changes in \(H_{3, i}\), except that
    the challenger samples uniformly random keys \(\dk'_{i}\) to compute the
    \(i\)-the component in $\ket{\qpk'}$, except for the one used to encrypt the
    challenge query.

    Similar to the previous argument, $H_{3, i+1}$ and $H_{3, i}$ are also
    computationally indistinguishable due to pseudorandomness property of
    $\ket{\psi_{\dk'_{i}, \cdot}}$.

  \item \textbf{Hybrid \(H_{4, i}\) for \(i \in [0, \secpar]\).}
    We define a series of (inefficient) hybrids \(H_{4, i}\), in which
    \(H_{4, 0} \coloneqq H_{4}\), and we denote
    \(H_{4, \secpar} \coloneqq H_{5}\).

    Each hybrid \(H_{4, i}\) is identical to \(H_{4, i+1}\), except that for the
    challenge encryption, the challenger does not sample \(\dk_{1, i}\) and
    compute \(\ket{\psi_{\dk_{1, i}, x^{*}_{i}}}\).
    Instead, the challenger generates \(\ket{\vartheta_{x^{*}_{i}}}\) using a
    Haar random oracle \(\ooracle_{\mathsf{Haar}_{i}}\) and uses this state to
    compute \(y'_{i}\) (by applying \(\destruct\) to
    \(\ket{\vartheta_{x^{*}_{i}}}\)).

    By the pseudorandomness of \(\ket{\psi_{\dk_{1, i}, \cdot}}\), \(H_{4, i}\) and
    \(H_{4, i+1}\) are computationally indistinguishable.

    \item \textbf{Hybrid \(H_{6}\).}
    This hybrid is identical to \(H_{5}\), except that now the challenger sets
    \(\tilde{y}_{i} = y_{i}\) for all \(i\) for the challenge encryption
    query.

    Note that in this hybrid, both \(y_{i}\) and \(y'_{i}\) are computed by
    applying \(\destruct\) to a Haar random state, thus they are output of the
    same distribution. Therefore, \(H_{5}\) and \(H_{6}\) are identical.

  \item \textbf{Hybrid \(H_{6, i}\) for \(i \in [0, \secpar]\).}
    We define a series of hybrids \(H_{6, i}\), in which
    \(H_{6, 0} \coloneqq H_{6}\), and we denote
    \(H_{6, \secpar} \coloneqq H_{7}\).

    Each hybrid \(H_{6, i+1}\) is identical to \(H_{6, i}\), except now instead
    of using a Haar random oracle in encryption of the challenge query, the
    challenger samples a fresh key \(\dk_{i}\) and uses this key to compute
    \(\tilde{y}_{i}\) which is a proof of destruction of the state
    \(\ket{\psi_{\dk_{i}, x^{*}_{i}}}\).

    By pseudorandomness of \(\ket{\psi_{\dk_{i}, \cdot}}\), \(H_{6, i+1}\) and
    \(H_{6, i}\) are computationally indistinguishable.

    We also note that the hybrid \(H_{7}\) is now efficient again.
    In this final hybrid, we note that the secret key \(k\) of the symmetric key
    encryption scheme is uniformly random and independent from all the other
    variables in the hybrid.
    Thus, we can easily reduce the winning probability of the adversary in this
    hybrid to the security of the symmetric key encryption scheme, which is
    negligible.
  \end{itemize}
  Overall, we obtain the winning probability of the adversary in the first
  hybrid \(H_{0}\) is negligible, and conclude the proof.

\end{proof}

\section{Impossibility of Unconditionally Secure qPKE}
\label{sec:impossibility}

In the following, we investigate the question on whether qPKE is possible to construct with information-theoretic security, and we give strong bounds against this. First, let us mention that a recent work by Morimae et al.~\cite{cryptoeprint:2022/1336} shows that an object called quantum pseudo-one-time pad (QPOTP) implies the existence of efficiently samplable, statistically far but computationally indistinguishable pairs of (mixed) quantum states (EFI pairs). QPOTP is a one-time symmetric encryption with quantum ciphertexts and classical keys, whose key length is shorter than the message length. qPKE immediately implies the existence of QPOTP, by increasing the message length, using bit-by-bit encryption. Since EFI pairs cannot exist information-theoretically, this chain of implications rules out the existence of unconditionally secure qPKE.\footnote{This observation was pointed out to us by Takashi Yamakawa.}

For the sake of completeness, we provide a new and direct proof of the impossibility statement using a shadow tomography argument.

\paragraph{A Proof from Shadow Tomography.}
In order to prove our impossibility result, we first show that if two public-keys $\ket\qpk$ and $\ket{\qpk^*}$ are close, if we encrypt a random bit using $\ket{\qpk^*}$, the probability of decrypting correctly with $\dk$ is high, where $\dk$ is the corresponding secret-key of $\ket{\qpk}$.

\begin{lemma}\label{lemma:pk_distance}
Let $\lambda$ be the security parameter and $\Gamma = (\Gen,\qpkgen,\qenc,\qdec)$ be a qPKE. Let $\dk^*,\ket{\qpk}^*$ be a fixed pair of honestly generated keys and for all decryption keys $\dk$ define $p_{\dk}$ to be:
 \begin{align*}
    p_{\dk} = \Pr\left[
\qdec(\dk, \qc) = \pt
\ \middle\vert
\begin{array}{r}
\pt \xleftarrow{\$} \{0,1\} \\
(\qc,\cdot) \gets \qenc(\qpk^* , \pt) 
\end{array}
\right]
\end{align*}

and let $\ket{\qpk_{\dk}} \gets \qpkgen(\dk)$. For all $\dk$, if $\abs{\langle \qpk^*
\vert \qpk_{\dk} \rangle} \geq 1- \epsilon$, then $p_{\dk}\geq 1-\sqrt{3\epsilon}$. 

\end{lemma}
\begin{proof}
Let $U_{\qenc}$ be the purified implementation of the encryption procedures, i.e. given the state $\ket{\qpk^*}\ket{b}\ket{0}$, $U_{\qenc}$ computes the state computed by $\qenc$ prior to the measurement. We argue that for any $\ket{\qpk_{\dk}}$ which is close to $\ket{\qpk^*}$, the purified ciphertexts generated by the two keys are also close. For any bit $b$, the purified ciphertext are defined as $\Tilde{\qc_b}=U_{\qenc}\ket{\qpk^*}\ket{b}\ket{0}\bra{0}\bra{b}\bra{\qpk^*}U_{\qenc}^\dagger$ and
$\Tilde{\qc_b}'=U_{\qenc}\ket{\qpk_{\dk}}\ket{b}\ket{0}\bra{0}\bra{b}\bra{\qpk_{\dk}}U_{\qenc}^\dagger$. We refer to these as purified ciphertexts. 
Now we can show, 

\begin{align}
    \trace(\Tilde{\qc_b}\Tilde{{\qc_b}}'^\dagger) &= \trace(U_{\qenc}\langle\qpk^* \vert \qpk_{\dk}\rangle \ket{\qpk^*}\bra{\qpk_{\dk}}U_{\qenc}^\dagger) \label{eq:trace1}\\ 
    &= \abs{\langle\qpk^* \vert \qpk_{\dk}\rangle}^2 \geq (1- \epsilon)^2\label{eq:trace2}
\end{align}
The transition from \Cref{eq:trace1} to \Cref{eq:trace2} follows from the trace-preserving property of unitaries. Let $\{\Pi^b_{\dk}\}_{\dk}$ be the POVM corresponding to decrypting a purified ciphertext with key $\dk$, i.e. the probability of a purified ciphertext $\qc$ being decrypted to $b$ by $\dk$ is given by $\trace(\Pi^b_{\dk}\qc)$. Now the term $p_{\dk}$ can be rewritten as follows:
\begin{align}
    p_{\dk} = \frac{1}{2}[\trace(\Pi^0_\dk\tilde{\qc}_0) + \trace(\Pi^1_\dk\tilde{\qc}_1)]
\end{align}

Now note that, $\trace(\Pi^0_\dk\qc_0')=\trace(\Pi^1_\dk\qc_1')=1-\negl$ as we assumed $\Gamma$ has negligible correctness error.  Now we can bound $p_{\dk}$ as follows, 
\begin{align}
    p_{\dk} &= \frac{1}{2}[\trace(\Pi^0_\dk\tilde{\qc}_0) + \trace(\Pi^1_\dk\tilde{\qc}_1)] \\
    &\geq 1-\negl - \frac{1}{2}[\trace(\vert\Pi^0_\dk(\tilde{\qc}_0-\tilde{\qc}_0')\vert) + \trace(\vert\Pi^1_\dk(\tilde{\qc}_1-\tilde{\qc}_1')\vert)]\\
    &\geq 1-\negl - \frac{1}{2}[\trace(\vert\tilde{\qc}_0-\tilde{\qc}_0'\vert) + \trace(\vert\tilde{\qc}_1-\tilde{\qc}_1'\vert)]\label{eq:ineq1}\\
    &=  1-\negl -\frac12 [\sqrt{1-\trace(\Tilde{\qc_0}\Tilde{{\qc_0}'}^\dagger)} + \sqrt{1-\trace(\Tilde{\qc_1}\Tilde{{\qc_1}'}^\dagger)}]\label{eq:ineq2}\\ 
    &\geq 1-\negl -  \sqrt{2\epsilon}\geq 1-\sqrt{3\epsilon}
\end{align}
The transition from \Cref{eq:ineq1} to \Cref{eq:ineq2} is due to $\tilde{\qc}_b$ and $\tilde{\qc}'_b$ being pure states. This concludes the proof of the lemma.

\end{proof}

Given \cref{lemma:pk_distance} one can reduce the adversary's task in the IND-CPA game to finding a decryption key $\dk$ such that the state
  $\ket{\qpk_{\dk}} \gets \qpkgen (\dk)$ is close to $\ket{\qpk^*}$ in inner product
  distance.
 The main technique we use to realize this subroutine of the adversary is shadow tomography
  introduced by Aaronson et al.~\cite{STOC:Aaronson18}.
  At the core of our proof is the following theorem by Huang, Kueng, and
  Preskill~\cite{huang2020predicting}.

\begin{theorem}[Theorem~1 and S16 \cite{huang2020predicting}]\label{thm:shadow_tmg}
Let $O_1,\dots,O_M$ be $M$ fixed observables and let $\rho$ be an unknown $n$-qubit state. Given $T = O(\log(M/\delta)/\epsilon^2\times \max_i \trace(O_i^2))$ copies of $\rho$, there exists a quantum algorithm that performs measurements in random Clifford basis on each copy and outputs $\tilde{p}_1,\dots,\tilde{p}_M$ such that, with probability at least $1-\delta$
$$\forall i , |\tilde{p}_i-\trace(O_i\rho)|\leq \epsilon$$
\end{theorem}

At a high level, the theorem states that outcomes of polynomially many random Clifford measurements on a state, i.e. a polynomial number of classical shadows, are enough to reconstruct an estimate of the statistics obtained by measuring an exponential number of observables. Note that, the post-processing required to reconstruct $\tilde{p}_i$ values is often inefficient, however for our purpose, i.e. proving the impossibility of an information-theoretically secure quantum PKE the efficiency of the procedure is not of concern.
Using \cref{thm:shadow_tmg} we are able to prove the impossibility statement.

\begin{theorem}
For any security parameter $\lambda$ and qPKE $\Gamma = (\Gen, \qpkgen,\qenc,\qdec)$ there exists a polynomial $m$ and a computationally unbounded adversary $\adv$ who can win the IND-CPA game with significant advantage only given $m(\secpar)$ copies of the public-key. 
\end{theorem}
\begin{remark}
Actually our attack allows us to recover the secret key with high probability, and thus the attack also breaks the one-wayness security of qPKE (which is a weaker security notion than IND-CPA). 
Thus, our theorem indeed shows a generic impossibility of unconditionally secure qPKE.
\end{remark}
\begin{proof}
   Let us describe the adversary given $m$ copies of the public-key $\ket{\qpk^*}$ alongside a challenge ciphertext $\qc$. We set the value of $m$ later in the proof. For a value $N$, we define the following rank 1 projection ensemble $\{\Pi^1_{\dk} = \ket{\qpk_{\dk}}\bra{\qpk_{\dk}}^{\otimes N} \}_{\dk \gets \Gen(1^\lambda)}$.
   The adversary tries to find a decryption key $\dk$ such that $\trace(\Pi^1_{\dk} \ketbra{\qpk^*}^{\ot N})$ is relatively large. 
   In order to do so the adversary computes $\trace(\Pi^1_{\dk} \ketbra{\qpk^*}^{\ot N})$ for all decryption keys $\dk$.
   By following the procedure from \cref{thm:shadow_tmg} on $\rho =\ket{\qpk^*}\bra{\qpk^*}^{\ot N}$, the adversary performs random Clifford measurements on \[T = O\left(\log\left(\frac{\#\{\dk|\dk\gets \Gen(1^\lambda)\}}{\delta}\right)\frac{1}{\epsilon^2} \trace({{\Pi_\dk^1}}^2)\right)\] copies of $\rho$ to compute values $\tilde{p}_\dk$ such that with probability $1-\delta$, for all $\dk$  \[\abs{\tilde{p}_\dk - \trace(\Pi^1_{\dk} \ketbra{\qpk^*}^{\ot N})} \leq \epsilon.\] Let us set $\epsilon < 1/6$ and $\delta$ to be a small constant, e.g. $1/100$. Immediately it can be noticed that as $\epsilon$ and $\delta$ are constants and ${\trace({\Pi_{\dk}}^1}^2)=1$\footnote{this is due to $\Pi^1_\dk$ operators being rank-1 projections}, $T$ is $O(\log(\#\{\dk|\dk \gets \gen(1^\secpar)\}))$ which is $\poly$ as the key-lengths should be polynomial in the security parameter.
   
   We claim that if the adversary picks any key such that $\tilde{p}_{\dk}>1/2$, they have found a key that has a high chance of decrypting the challenge ciphertext correctly. Let us elaborate. First of all, note that the adversary finds at least one such $\dk$ with probability at least $1-\frac{1}{100}$, as for the correct decryption key $\dk^*$, $\trace(\Pi^1_{\dk^*} \ketbra{\qpk^*}^{\ot N}) = 1$ hence $\tilde{p}_{\dk^*} > 1-1/6$ with probability at least $1-\frac{1}{100}$. 

    The next thing to show is that any $\dk$ such that $\tilde{p}_{\dk} > 1/2$ is a \emph{good} decryption key. Note that due to \cref{lemma:pk_distance} we have,
    \begin{align}
        \trace(\Pi^1_{\dk} \ketbra{\qpk^*}^{\ot N}) = \abs{\langle\qpk_\dk \vert \qpk^* \rangle}^{2N}
    \end{align}
    We note that for all $\dk$ such that $p_{\dk} \leq 1-\sqrt{\frac{3}{\log(N)}}$ we have: 
    \begin{align}
        p_{\dk} &\leq 1-\sqrt{\frac{3}{\log(N)}} \Rightarrow \langle\qpk_\dk \vert \qpk^* \rangle \leq  1- \frac{1}{\log(N)} \\ 
        &\Rightarrow \trace(\Pi^1_{\dk} \ketbra{\qpk^*}^{\ot N}) \leq (1-\frac{1}{\log(N)})^{2N}\\
        &\leq e^{-2N/\log(N)} \ll 1/3 \text{,  for a large enough }N
    \end{align}

Given our choice of $\delta,\epsilon$, this ensures that if the adversary picks any $\dk$ such that $\tilde{p}_\dk> 1/2$, with probability at least $1-\frac{1}{100}$ we have that, $\abs{\tilde{p}_{\dk} - \trace(\Pi^1_{\dk} \ketbra{\qpk^*}^{\ot N})} \leq 1/6$, $\trace(\Pi^1_{\dk} \ketbra{\qpk^*}^{\ot N})> 1/3$ hence, $p_{\dk}> 1-\sqrt{\frac{3}{\log(N)}}$. 

As the last step, the adversary uses the $\dk$ they obtain from the previous procedure to decrypt the challenge ciphertext $\qc^*$. By union bound and following the discussion above the adversary's advantage to decrypt the challenge ciphertext correctly is greater than $1-\frac{1}{100}-\sqrt{\frac{3}{\log(N)}}$ which by setting $N$ to be a large constant is significantly larger than $1/2$. Finally note that this adversary uses $m= NT$ copies of the public-key, where $T=\poly$ and $N$ is a constant, so the total number of public-key copies used are polynomial in $\secpar$.
\end{proof}
 
\ifanon
\section*{Acknowledgments}
The authors wish to thank Takashi Yamakawa for pointing out a simple argument to rule out the existence of information-theoretically secure qPKE. The argument is replicated here with his permission.
\else
\section*{Acknowledgments}
The authors wish to thank Prabhanjan Ananth and Umesh Vazirani for related discussions, and Takashi Yamakawa for pointing out a simple argument to rule out the existence of information-theoretically secure qPKE. The argument is replicated here with his permission.

ABG and QHV are supported by ANR JCJC TCS-NISQ ANR-22-CE47-0004, and by the PEPR
integrated project EPiQ ANR-22-PETQ-0007 part of Plan France 2030.
GM was partially funded by the German Federal Ministry of Education and Research
(BMBF) in the course of the 6GEM research hub under grant number 16KISK038 and
by the Deutsche Forschungsgemeinschaft (DFG, German Research Foundation) under
Germany's Excellence Strategy - EXC 2092 CASA – 390781972.
OS was supported by the Israeli Science Foundation (ISF) grant No.
682/18 and 2137/19, and by the Cyber Security Research Center at Ben-Gurion
University.
KB and LH are supported by the Swiss National Science Foundation (SNSF) through
the project grant 192364 on Post Quantum Cryptography.

\BeforeBeginEnvironment{wrapfigure}{\setlength{\intextsep}{0pt}}
\begin{wrapfigure}{r}{90px}
  \includegraphics[width=40px]{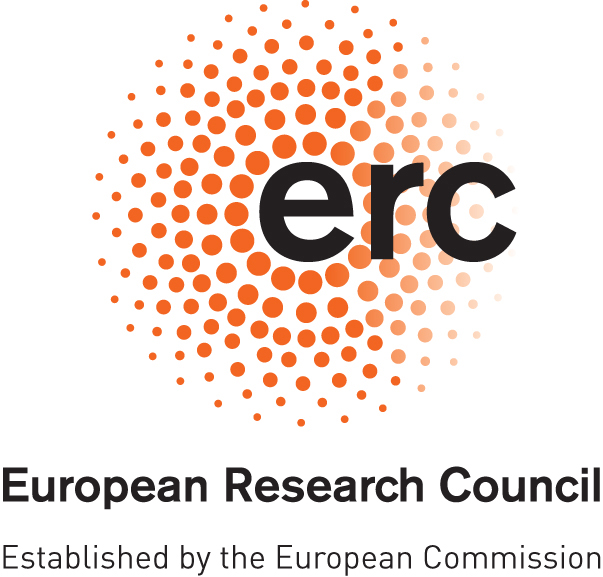}
  \includegraphics[width=40px]{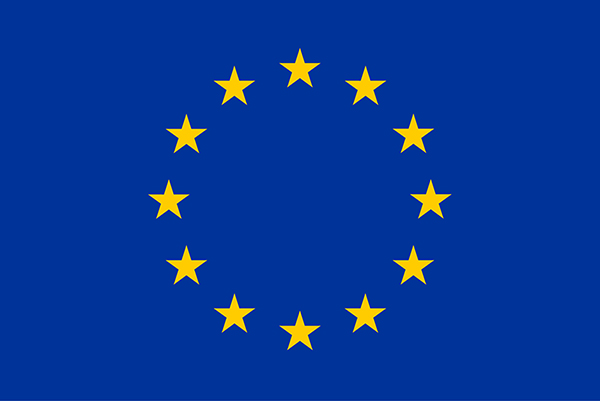}
\end{wrapfigure}
OS has received funding from the European Research Council (ERC) under the
European Union’s Horizon 2020 research and innovation programme (grant agreement
No 756482).
MW acknowledges support by the the European Union (ERC, SYMOPTIC, 101040907), by
the Deutsche Forschungsgemeinschaft (DFG, German Research Foundation) under
Germany's Excellence Strategy - EXC\ 2092\ CASA - 390781972, by the BMBF through
project QuBRA, and by the Dutch Research Council (NWO grant OCENW.KLEIN.267).
Views and opinions expressed are those of the author(s) only and do not
necessarily reflect those of the European Union or the European Research Council
Executive Agency.
Neither the European Union nor the granting authority can be held responsible
for them.

\fi
\bibliographystyle{alpha}
\renewcommand{\doi}[1]{\url{#1}}
\bibliography{./cryptobib/abbrev3,./cryptobib/crypto,main}

\appendix
\crefalias{section}{appendix}
\section{CCA-Secure Bit-Encryption from OWF}
\label{app:CCA_bit_encryption}

In this appendix, we describe a simple quantum public key bit encryption scheme that satisfies the strong notion of CCA security.
The construction relies on a quantum-secure pseudorandom function
\[
  \PRF\colon \{0,1\}^\lambda \times \{0,1\}^\lambda \to \{0,1\}^{3\lambda}
\]
which, as mentioned earlier in \cref{sec:definitions}, can be constructed from
any quantum-secure one-way function.
Then our quantum PKE scheme~$\Pi=(\Gen,\qpkgen,\qenc,\qdec)$ is defined as
follows:
\begin{itemize}
\item The key generation algorithm $\Gen(1^\lambda)$ samples two keys $\dk_0 \xleftarrow{\$} \{0,1\}^\lambda$ and $\dk_1 \xleftarrow{\$} \{0,1\}^\lambda$ and sets $\dk = (\dk_0,\dk_1)$. The public-key generation $\qpkgen(\dk)$ prepares the states
\begin{align*}
  \ket{\qpk_0} = \sum_{x\in \{0,1\}^{\lambda}} \ket{x,f_{\dk_0}(x)}
\quad\text{and}\quad
  \ket{\qpk_1} = \sum_{x\in \{0,1\}^{\lambda}} \ket{x,f_{\dk_1}(x)}.
\end{align*}
    Where $\{f_{\dk}\}_{\dk}$ is a $\PRF$.
    Note that both states are efficiently computable since the $\PRF$ can be
    efficiently evaluated in superposition in view of \cref{thm:quantum}.
    The quantum public key is then given by the pure state
    $\ket\qpk = \ket{\qpk_0} \ot \ket{\qpk_1}$, whereas the classical secret key
    consists of the pair $\dk = (\dk_0,\dk_1)$.

\item Given a message $\pt \in \{0,1\}$, the encryption algorithm $\qenc(\ket\qpk, \pt)$ simply measures $\ket{\qpk_{\pt}}$ in the computational basis, and outputs the measurement outcome as the \emph{classical} ciphertext $\qc=(x,y)$ and the post measurement state $\ket{x}\ket{y}$.

\item Given the ciphertext $\qc=(x,y)$, the decryption algorithm $\qdec(\dk, \qc)$ first checks whether $f_{\dk_0}(x) = y$ and returns $0$ if this is the case.
Next, it checks whether $f_{\dk_1}(x) = y$ and returns $1$ in this case.
Finally, if neither is the case, the decryption algorithm returns $\bot$.
\end{itemize}

Next, we establish correctness of this scheme.

\begin{theorem}\label{thm:correctness}
If $\PRF$ is a quantum-secure pseudorandom function, then the quantum PKE scheme $\Pi$ is correct.
\end{theorem}
\begin{proof}
Observe that the scheme is perfectly correct if the ranges of $f_{\dk_0}$ and $f_{\dk_1}$ are disjoint.
By a standard argument, we can instead analyze the case of two truly random functions~$f_0$ and~$f_1$, and the same will hold for~$f_{\dk_0}$ and~$f_{\dk_1}$, except on a negligible fraction of the inputs.
Fix the range of~$f_0$, which is of size at most~$2^\lambda$.
Then the probability that any given element of~$f_1$ falls into the same set is at most $2^{-2\lambda}$, and the desired statement follows by a union bound.
\end{proof}

Finally, we show that the scheme is CCA-secure. The main tool used in the proof is the one-way to hiding lemma~\cite{C:AmbHamUnr19}.

\begin{lemma}[One-way to hiding]\label{lemma:o2h}
  Let $G,H: X\to Y$ be random functions and $S \subset X$ an arbitrary set with the
  condition that $\forall x\notin S, G(x) = H(x)$, and let $z$ be a random bitstring.
  Further, let $\adv^H(z)$ be a quantum oracle algorithm that queries $H$ with
  depth at most $d$.
  Define $\bdv^H(z)$ to be an algorithm that picks $i\in[d]$ uniformly, runs
  $\adv^H(z)$ until just before its $i^{th}$ round of queries to $H$ and
  measures all query input registers in the computational basis and collects
  them in a set $T$.
  Let
\begin{align*}
  P_{\text{left}} = \Pr[ 1\gets \adv^H(z)], \quad
  P_{\text{right}} = \Pr[1\gets \adv^G(z)], \\
  P_{\text{guess}} = \Pr[S\cap T \neq \emptyset | T\gets \bdv^H(z)]
\end{align*}
Then we have that
\begin{align}
    |P_{\text{left}} - P_{\text{right}}| \leq 2d \sqrt{P_{\text{guess}}} \quad\text{and}\quad |\sqrt{P_{\text{left}}} - \sqrt{P_{\text{right}}}| \leq 2d \sqrt{P_{\text{guess}}}
\end{align}
\end{lemma}

\begin{theorem}\label{thm:CCA}
If $\{f_{\dk}\}_{\dk}$ is a quantum-secure pseudorandom function ensemble, then the quantum PKE scheme $\Pi$ is CCA-secure.
\end{theorem}
\begin{proof}
It suffices to show that the CCA experiment with the bit~$b$ fixed to~$0$ is indistinguishable from the same experiment but with~$b$ fixed to~$1$.
To this end we consider a series of hybrids, starting with the former and ending with the latter:
\begin{itemize}
    \item \textbf{Hybrid 0:} This is the original CCA experiment except that the bit $b$ fixed to $0$.
    \item \textbf{Hybrid 1:} In this (inefficient) hybrid, we modify hybrid~0 to instead compute $\ket{\qpk_0}$ as
    \[
    \ket{\qpk_0} = \sum_{x\in \{0,1\}^{\lambda}} \ket{x,f(x)},
    \]
    where $f$ is a truly uniformly random function.
\end{itemize}
The indistinguishability between these two hybrids follows by a standard reduction against the quantum security of~$\PRF$:
To simulate the desired $n$ copies of $\ket{\qpk_0}$, and to answer decryption queries (except the one that contains the challenge ciphertext), the reduction simply queries the oracle provided by the $\PRF$ security experiment (possibly in superposition). Note that whenever the oracle implements $\PRF$, then the view of the distinguisher is identical to hybrid $0$, whereas if the oracle implements a truly random function, then the view of the distinguisher is identical to hybrid $1$.

\begin{itemize}
  \item \textbf{Hybrid 2:} In this (inefficient) hybrid, we modify hybrid~1 such
    that the challenge ciphertext is sampled as
    \[
    x \xleftarrow{\$} \{0,1\}^{\lambda} \quad\text{and}\quad y \xleftarrow{\$} \{0,1\}^{3\lambda}.
    \]
\end{itemize}
The indistinguishability of hybrids~1 and~2 follows from the one-way to hiding
lemma (\cref{lemma:o2h}).
Let $H$ be such that $H(x) = y$ and for all $x'\neq x$ we set $H(x') = f(x')$, and
let $S = \{x\}$.
Let $\adv$ be the adversary playing the security experiment.
We claim that $\adv^f$ is the adversary playing in hybrid~$1$ whereas $\adv^H$
corresponds to the adversary playing hybrid~$2$: Observe that the public keys
can be simulated with oracle access to $f$ ($H$, respectively) by simply
querying on a uniform superposition of the input domain, whereas the decryption
queries can be simulated by query basis states.
Importantly, for all queries after the challenge phase, the adversary is not
allowed to query $x$ to $\qdec^*$.
Hence the set $T$, collected by $\bdv$ is a set of at most $n$ uniform elements
from the domain of $f$, along with $Q$ basis states, where $Q$ denotes the
number of queries made by the adversary to the decryption oracle \emph{before}
the challenge ciphertext is issued.
By a union bound
\[P_{\text{guess}} =\Pr[T \cap \{x\} \neq \emptyset] \leq \frac{(n+Q)}{2^\lambda}=\textsf{negl}(\lambda)\]
since $x$ is uniformly sampled.
Applying~\cref{lemma:o2h}, we deduce that $|P_{\text{left}}-P_{\text{right}}|$
is also negligible, i.e., which bounds the distance between the two hybrids.

\begin{itemize}
    \item \textbf{Hybrid 3:} In this (efficient) hybrid, we modify hybrid~2 to compute $\ket{\qpk_0}$ by using the pseudorandom function~$f_{\dk_0}$ instead of the truly random function~$f$.
    That is, we revert the change done in hybrid~$1$.
\end{itemize}
Indistinguishability follows from the same argument as above.

\begin{itemize}
    \item \textbf{Hybrid 4:} In this (inefficient) hybrid, we modify hybrid~3 to compute $\ket{\qpk_1}$ as
    \[
    \ket{\qpk_1} = \sum_{x\in \{0,1\}^{\lambda}} \ket{x,f(x)}
    \]
    where $f$ is a truly uniformly random function.
\end{itemize}
Indistinguishability follows from the same argument as above.

\begin{itemize}
    \item \textbf{Hybrid 5:} In this (inefficient) hybrid, we modify hybrid~4 by fixing the bit $b$ to $1$ and computing the challenge ciphertext honestly, i.e., as
        \[
    x \xleftarrow{\$} \{0,1\}^{\lambda} \quad\text{and}\quad y =f(x).
    \]
\end{itemize}
Indistinguishability follows from the same argument as above.

\begin{itemize}
    \item \textbf{Hybrid 6:} In this (efficient) hybrid, we modify hybrid~5 to compute $\ket{\qpk_1}$ by using the pseudorandom function~$f_{\dk_1}$ instead of the truly random function~$f$.
    That is, we revert the change done in hybrid~$4$.
\end{itemize}
Indistinguishability follows from the same argument as above. The proof is concluded by observing that the last hybrid is identical to the CCA experiment with the bit~$b$ fixed to $1$.
\end{proof}
 \end{document}